\documentclass[letterpaper]{article} 
\usepackage{aaai2026}  
\usepackage{times}  
\usepackage{helvet}  
\usepackage{courier}  
\usepackage[hyphens]{url}  
\usepackage{graphicx} 
\usepackage{natbib}  
\usepackage{caption} 
\usepackage[ruled,vlined]{algorithm2e}

\usepackage{amsmath}
\usepackage{amsfonts}
\usepackage{amssymb}
\usepackage{amsthm}

\usepackage{multirow}
\usepackage{booktabs}
\usepackage{xcolor}
\usepackage{makecell}
\usepackage{csquotes}
\newtheorem{theorem}{Theorem}[section]
\newtheorem{lemma}[theorem]{Lemma}
\newtheorem{corollary}[theorem]{Corollary}
\newtheorem{proposition}[theorem]{Proposition}

\DeclareCaptionStyle{ruled}{labelfont=normalfont,labelsep=colon,strut=off} 

\title{\method: A Realistic Black-Box Node Injection Attack on LLM-Enhanced GNNs}
\author {
    Jiaji Ma\textsuperscript{\rm 1}\thanks{Corresponding Author},
    Puja Trivedi\textsuperscript{\rm 1},
    Danai Koutra\textsuperscript{\rm 1}
}
\affiliations {
    \textsuperscript{\rm 1}University of Michigan, Ann Arbor\\
    jiaji@umich.edu, pujat@umich.edu, dkoutra@umich.edu
}

\begin{document}

\newcommand{\method}{\textsc{GraphTextack}}
\newcommand{\matX}{\mathbf{X}}

\maketitle

\begin{abstract}
Text-attributed graphs (TAGs), which combine structural and textual node information, are ubiquitous across many domains. Recent work integrates Large Language Models (LLMs) with Graph Neural Networks (GNNs) to jointly model semantics and structure, resulting in more general and expressive models that achieve state-of-the-art performance on TAG benchmarks.
However, this integration introduces dual vulnerabilities: GNNs are sensitive to structural perturbations, while LLM-derived features are vulnerable to prompt injection and adversarial phrasing. While existing adversarial attacks largely perturb structure or text independently, we find that uni-modal attacks cause only modest degradation in LLM-enhanced GNNs. Moreover, many existing attacks assume unrealistic capabilities, such as white-box access or direct modification of graph data.
To address these gaps, we propose \method, the first black-box, multi-modal{, poisoning} node injection attack for LLM-enhanced GNNs. \method{} injects nodes with carefully crafted structure and semantics to degrade model performance, operating under a realistic threat model without relying on model internals or surrogate models. To navigate the combinatorial, non-differentiable search space of connectivity and feature assignments, \method{} introduces a novel evolutionary optimization framework with a multi-objective fitness function that balances local prediction disruption and global graph influence. Extensive experiments on five datasets and two state-of-the-art LLM-enhanced GNN models show that \method{} significantly outperforms 12 strong baselines.
\end{abstract}

\section{Introduction}
Text-attributed graphs (TAGs), where nodes are associated with natural language text, are common in many real-world applications. Examples include citation networks, where each node represents a paper described by its abstract, and product co-purchase graphs, where products are annotated with textual descriptions. TAGs combine rich semantic information with relational structure, making them a powerful representation for a range of learning tasks.
Early graph learning methods relied on shallow node features, often ignoring the rich semantic and contextual information in textual attributes. Recent advances leverage Large Language Models (LLMs) to extract expressive node representations from raw text, which are then integrated with graph structure via Graph Neural Networks (GNNs). This combination enables models to capture complex relational and semantic patterns, leading to improved performance on downstream tasks involving TAGs~\cite{duan2023simteg,tang2024graphgpt,fatemi2023talk}.

Compared to directly applying LLMs to graph learning tasks (i.e., the LLM-as-predictor paradigm), combining LLM-derived embeddings with graph structure aggregation~\cite{kipf2016semi,hamilton2017inductive,wu2020comprehensive}, has achieved state-of-the-art performance on node classification benchmarks~\cite{chien2021node,duan2023simteg,he2023harnessing,liu2023one}.
These LLM-enhanced GNN approaches offer a powerful and general framework for learning on TAGs, effectively integrating language modeling with structural reasoning.

However, this integration also introduces new vulnerabilities {that are underexplored}. GNNs are known to be susceptible to adversarial perturbations in both structure and node features~\cite{zugner2018adversarial,ma2021graph}, with structural modifications often being particularly effective~\cite{Zhu22_HeterophilyRobustness}. At the same time, LLMs are vulnerable to prompt injection and adversarial phrasing~\cite{Li24_BackdoorLLM,Wei23_Jailbroken}. As a result, LLM-enhanced GNNs inherit vulnerabilities from both modalities (structural and semantic) making them sensitive to even small perturbations in either input space.

Yet, because these vulnerabilities are distributed across two interdependent modalities, they are difficult to fully exploit through uni-modal attacks that perturb structure or text independently. To motivate the need for stronger attack strategies, we briefly evaluate state-of-the-art uni-modal attacks, including structural perturbations~\cite{geisler21_PRBCD} and textual attacks~\cite{Li20_BertAttack,Ebrahimi18_HotFlip,Eger19_VIPER}. Our observations show that uni-modal attacks cause only \textit{modest} degradation, highlighting the importance of developing effective multi-modal adversarial approaches.

Existing adversarial attacks on LLM-enhanced GNNs, however, largely remain \emph{uni-modal}, perturbing either structure or textual features independently~\cite{guo2024learning}. Moreover, many existing attacks assume unrealistic capabilities, such as direct modification of existing nodes, edges, or textual content, or full white-box access to model internals. In practical deployment settings, however, existing graph data is often immutable, and deployed models expose only prediction outputs without revealing internal architectures or training dynamics. These constraints make such strong attack assumptions infeasible for real-world adversaries. 

This motivates a more realistic threat model: \emph{node injection attacks}, where adversaries introduce new nodes into the graph without modifying the existing graph elements. This better reflects common real-world scenarios where adversaries can create new user profiles, products, or documents, but cannot arbitrarily alter existing records. 
Such attacks are especially relevant in domains like e-commerce (where  fake products or reviews may be injected to manipulate recommendations) and citation networks, where low-quality or fraudulent publications can be added to influence scholarly metrics. 
In these settings,  attackers can realistically create new entities with realistic structure and textual content, {which are added to the data and can influence the model,} making node injection a practical adversarial attack.

\noindent \textbf{Our Approach.} 
To address existing limitations, we propose \method, the first black-box, multi-modal poisoning attack that injects nodes, designed specifically for LLM-enhanced GNNs. It jointly perturbs both structure and semantics without relying on model internals or surrogate models. 
To efficiently explore the combinatorial search space of injections, we adopt an evolutionary optimization framework~\cite{eiben2003introduction}. 
This approach is well-suited for our setting because the search space includes both structural and semantic changes, which makes standard optimization difficult. Moreover, gradient-based methods are not applicable: LLM-enhanced GNNs rely on LLM encoders that produce high-dimensional, non-differentiable embeddings. Surrogate models are also ineffective, as they cannot accurately approximate the LLM-based representations.

In contrast to standard evolutionary attacks that operate solely on structural perturbations under surrogate-assisted settings~\cite{li2023lapa,fang2024gani}, \method{} systematically adapts the evolutionary optimization framework to the multi-modal nature of the node injection problem, and jointly optimizes both structure and semantics. 
Our approach introduces a joint candidate encoding, multi-modal crossover and mutation operations, and a novel multi-objective fitness function that balances structural and semantic attack objectives, enabling more effective and realistic attacks under practical threat models.

We summarize our main contributions as follows:
\begin{itemize}
    \item \textbf{Realistic, Multi-modal Attack Model.} We introduce \method, a realistic, black-box, multi-modal node injection attack that jointly optimizes structure and semantic features to degrade LLM-enhanced GNNs.
    \item \textbf{Theoretical Analysis.} 
    We provide a search space and complexity analysis supporting the efficiency and scalability of our approach, offering theoretical justification.
    \item \textbf{Empirical Study.} We evaluate \method{} across 5 diverse node classification benchmarks and 2 representative LLM-enhanced GNN target models, demonstrating its superior effectiveness and efficiency compared to 12 state-of-the-art baselines. 
\end{itemize}

\section{Related Work}
\setlength{\tabcolsep}{1mm}
\begin{table}[t!]
\centering
{\fontsize{9}{10}\selectfont
\begin{tabular}{lccc}
\toprule
\textbf{Method} & \textbf{Black-box} & \textbf{Multi-modal} & \textbf{Node  Injection} \\
\midrule
BertAttack& \checkmark & \makecell{\textbf{x} (text)} & {\textbf{x} } \\
HotFlip& \checkmark & \makecell{\textbf{x} (text)} & {\textbf{x} } \\
VIPER&  \checkmark & \makecell{\textbf{x} (text)} & {\textbf{x} } \\
PRBCD& \textbf{x} (white-box) & \textbf{x} (edges only)& \textbf{x} \\
Nettack& \textbf{x} (surrog.)  & \textbf{x} (graph)  & \textbf{x}  \\
PSO& \checkmark & \textbf{x} (graph) & \checkmark \\
TDGIA & \checkmark & \textbf{x} (graph) & \checkmark \\
AFGSM & \textbf{x} (surrog.) & \textbf{x} (graph) & \checkmark \\
$G^2A2C$ & \checkmark & \textbf{x} (graph) & \checkmark \\
GANI& \textbf{x} (surrog.) & \textbf{x} (graph) & \checkmark \\
WTGIA& \textbf{x} (surrog.) & \checkmark & \checkmark \\
\midrule 
Ours
& \checkmark & \checkmark{} (graph + text) & \checkmark \\
\bottomrule
\end{tabular}
}
\caption{{Qualitative comparison with prior attacks.}}
\label{tab:attack_comparison}
\end{table}

Our work builds on several lines of research on graph adversarial attacks, node injection strategies, and robustness of LLM-enhanced GNNs. To contextualize our contributions, we now review the most relevant prior work and highlight how \method{} advances the state of the art  in the context of realistic, multi-modal, and black-box node injection attacks. {A qualitative comparison of prior attack methods and our approach is provided in Table \ref{tab:attack_comparison}}. 

\textbf{Graph Adversarial Attacks.} 
\textit{Structural modifications.}
Adversarial attacks on graphs have traditionally focused on perturbing node features or graph structure to degrade model performance~\cite{ma2021graph, wei2020adversarial}. 
Among these, structural attacks have received greater attention and generally prove more effective, due to  decreasing homophily, amplification effects through message passing, and the creation of inconsistent neighborhoods~\cite{Zhu22_HeterophilyRobustness, Zhu20_H2GCN}. 
Attack strategies range from greedy approximations~\cite{zugner2018adversarial}, to meta-learning~\cite{zugner2020adversarial}, and evolutionary search methods~\cite{li2023lapa}.
However, these approaches typically assume direct access to and control over the existing graph structure, an assumption often unrealistic in practice, where attackers must operate without the ability to modify existing nodes or edges.

\textit{Node injection attacks.}
To overcome this limitation, node injection attacks introduce new nodes connected to the existing graph without modifying existing data.
AFGSM~\cite{wang2020scalable_AFGSM} applies an approximate fast gradient sign method to generate targeted perturbations; G-NIA~\cite{tao2021single_GNIA} proposes a parametric poisoning strategy that preserves learned structural patterns; TDGIA~\cite{zou2021_tdgia} is a targeted evasion attack based on topological defective edge selection mechanisms; and GANI~\cite{fang2024gani} uses a genetic algorithm to optimize node injections via a surrogate model.
Unlike these methods, \method{} operates fully black-box and jointly optimizes both structural connectivity and textual feature assignments, without relying on surrogate models or white-box assumptions.

\textbf{Attacks on LLM-Enhanced GNNs.} The robustness of LLM-enhanced GNNs remains underexplored. Existing adversarial strategies for LLM-enhanced GNNs largely remain uni-modal, targeting either structural or textual perturbations independently. Structural perturbations exploit vulnerabilities in message passing, while textual perturbations specifically affect the LLM component and the resulting semantic embeddings.
Previous work~\cite{guo2024learning} study adversarial attacks on LLM-as-enhancer and LLM-as-predictor pipelines but primarily consider white-box \textit{modification} attacks, i.e., uni-modal attacks where the edges of the original graph are perturbed.
On the other hand, WTGIA~\cite{lei2024intruding_WTGIA} introduces a text-only node injection attack, without considering joint structural and semantic vulnerabilities. 

Natural language adversarial attacks such as BERT-Attack~\cite{Li20_BertAttack}, HotFlip~\cite{Ebrahimi18_HotFlip}, and VIPER~\cite{Eger19_VIPER} were originally developed to degrade LLM performance through minimal text perturbations~\cite{Sclar24_WorryingAboutFormatting, Mizrahi24_StateOfWhatArt}, and can be adapted to attack the text modality in LLM-enhanced GNNs.

Unlike prior work on uni-modal attacks, \method{} introduces the first black-box, global node injection attack that jointly optimizes structural and semantic modalities, demonstrating superior effectiveness over uni-modal baselines across multiple benchmark dataset in realistic threat settings. {Notably, it does not require model gradients, surrogates, or data modification. This makes \method{} uniquely suited to realistic attack scenarios. }

\section{Preliminaries}
\label{sec:preliminaries}
In this section, we introduce key concepts for our work.

\noindent \textbf{Graph Definitions.}
We consider a text-attributed graph, or TAG, $G = (V, E, T, Y)$, where $V$ denotes the set of nodes, $E$ the set of edges, $T: V \to \mathcal{T}$ a mapping from each node to an associated textual attribute, and $Y \in \mathcal{Y}^{|V|}$ the node labels. 
Each model defines an encoding function $\phi: \mathcal{T} \to \mathbb{R}^d$ that processes the raw text into node feature embeddings. Thus, the node feature matrix is $\matX = [\phi(T(v))]_{v \in V} \in \mathbb{R}^{|V| \times d},$
where $\phi$ may correspond to frozen LLM encoders, task-specific fine-tuned models, or other text representation methods depending on the application.

\noindent \textbf{LLM-Enhanced GNNs.}
Recent approaches have integrated LLMs into GNNs to improve representation learning on TAGs~\cite{ye2023natural, tang2024graphgpt, liu2023one}.
In LLM-enhanced GNNs, textual attributes $T(v)$ are first processed through an LLM encoder to generate rich semantic embeddings, which are then aggregated using graph message-passing architectures. 
This design enables the model to leverage both semantic and structural information for downstream tasks. Empirical studies show that they often achieve accuracy surpassing traditional GNNs operating on shallow embeddings~\cite{liu2023one}.

\noindent \textbf{Adversarial Attack Settings.}
Uni-modal graph adversarial attacks seek to create \enquote{imperceptible} perturbations that deteriorate model performance. 
These attacks can be categorized as \textit{evasion} attacks, which perturb the graph data at test time, and \textit{poisoning} attacks, which perturb the graph data during training. Attackers may target specific nodes (\textit{targeted} attacks), or aim to reduce overall model performance (\textit{global} attacks). 
In the following sections, we consider several adversarial attack scenarios targeting LLM-enhanced GNNs: (1) \textbf{Textual perturbation attacks}: adversarial modifications to the raw text $T(v)$ associated with existing nodes, without altering the graph structure;
(2) \textbf{Structural perturbation attacks}: adversarial modifications to the edge set $E$ (e.g., adding or removing edges between existing nodes), without changing node features; 
(3) \textbf{Node injection attacks}: injection of new nodes $V'$ with associated edges $E'$ and features derived from text $T'(v')$.

In textual and structural attacks, adversaries directly modify $T$ or $E$, resulting in perturbed graphs $G_{\text{text}} = (V, E, T^{\text{adv}}, Y)$ or $G_{\text{struct}} = (V, E^{\text{adv}}, T, Y)$, where $T^{\text{adv}}$ and $E^{\text{adv}}$ denote adversarially modified text and edge sets.

In the node injection setting, the attacker constructs an augmented graph 
$G' = (V \cup V', E \cup E', T \cup T', Y \cup Y'),$
where $T'$ and $Y'$ denote the text and labels associated with the injected nodes, respectively. The goal is to inject nodes such that the model's predictions are significantly disrupted, without modifying existing nodes or edges. 
Throughout the paper, $G$ refers to the clean graph, $G'$ to the graph after node injection, and $f$ to the target model under attack.

\section{{Adversarial Setting and Threat Model}}
\label{sec:method}
LLM-enhanced GNNs leverage language models to enrich node representations with semantic information extracted from textual descriptions. While this added expressivity improves performance, it also introduces new vulnerabilities: perturbations to LLM-derived features or their interactions with the graph structure can significantly influence neighborhood aggregation and classification outcomes. 

\noindent \textbf{Motivation: Limitations of uni-modal attacks.} Recent work has begun to study adversarial attacks against TAGs in LLM-GNN joint models~\cite{lei2024intruding_WTGIA, guo2024learning}, typically treating structure and text as disjoint modalities. As a result, such studies evaluate structural attacks and textual attacks separately. However, structure and semantic features are inherently entangled, and real-world adversaries may exploit both modalities simultaneously.

To better understand these vulnerabilities, we first empirically evaluate the robustness of LLM-enhanced GNNs to uni-modal poisoning attacks, including textual perturbations (BERT-Attack~\cite{Li20_BertAttack}, 
HotFlip~\cite{Ebrahimi18_HotFlip}, VIPER~\cite{Eger19_VIPER}) and structural perturbations (random edge modifications, PRBCD~\cite{zugner2018adversarial}). 
We discuss the detailed empirical setup in Sec.~\ref{sec:experiments} and the detailed results in App.~\ref{app:unimodal-results}, and only give the key observations that motivate our approach here. 
Uni-modal attacks result in \textit{only modest} performance degradation: textual attacks typically reduce accuracy by less than 5\%, and structural attacks require randomly altering over 10\% of edges or white-box access to achieve comparable effects. 
These findings highlight the need for stronger, more realistic attacks---particularly ones that simultaneously manipulate structure and semantic features, without modification of existing graphs. 

\noindent \textbf{Problem Setup.} Unlike prior uni-modal perturbations, we focus on a more restricted and practical threat model: injecting new nodes without altering existing nodes or edges. 
Under our \textbf{threat model}, the attacker can inject a set of adversarial nodes $V'$, with associated edges $E'$ and textual descriptions $T'$, leading to an augmented graph $G' = (V \cup V', E \cup E', T \cup T', Y \cup Y')$. 

We study a \textit{poisoning attack} setting, where adversarial nodes are injected prior to model training to influence the learning process. The attacker seeks to strategically design $(V', E', T')$ to maximally degrade the model's predictive performance on the unlabeled nodes, while operating under realistic constraints.
Specifically, we consider a \textit{black-box attack} setting: the attacker has no access to the model's parameters, gradients, or internal architecture, and can only observe model predictions through queries. 
{In practice, these queries are lightweight inferences, and results can be cached to avoid redundant calls. }
This mirrors real-world adversarial scenarios where deployed models are typically opaque beyond input-output access. 
We focus on \emph{global attacks} aimed at degrading the model's overall performance across the graph, rather than targeting specific nodes individually.

\section{{\method: Our Attack Framework}}
\label{sec:theory}
To effectively search the space of possible node injections under the black-box threat model defined in Sec.~\ref{sec:method}, we propose \method, a black-box evolutionary optimization framework designed for node injection poisoning attacks on text-attributed graphs. 

\textbf{Motivation for an evolutionary attack.} Optimizing node injections under a black-box model is uniquely challenging: the search space is high-dimensional, discrete, combinatorial, and non-differentiable. Traditional black-box optimization techniques, such as Bayesian optimization, struggle due to reliance on surrogate models. In contrast, evolutionary algorithms (EAs) offer a pragmatic, assumption-light approach, requiring only model queries and can explore complex solution spaces without gradient information.

While EAs have been applied in adversarial machine learning, prior work primarily targeted structural perturbations under white-box or surrogate-assisted settings. \method{} addresses a significantly harder problem: realistic, global node injection attacks on LLM-enhanced GNNs, where injected nodes influence both graph structure and LLM-derived features. By jointly evolving structural connections and feature assignments, \method{} effectively navigates the combinatorial complexity introduced by multi-modal vulnerabilities, outperforming prior approaches constrained to single-modality attacks.

\begin{algorithm}[!htb]
\caption{\method}
\label{alg:graphtextack}
\KwIn{Graph $G = (V, E, T, Y)$, model $f$, injection budget $r$}
\KwOut{Poisoned graph $G' = (V \cup V', E \cup E', T', Y')$}

\For{$t = 1$ to $r$}{
    Initialize population $\mathcal{P}_0 = \{s_i\}_{i=1}^{N_p}$\;
    \ForEach{candidate $s_i \in \mathcal{P}_0$}{
        Randomly initialize structural connections and feature generation strategy\;
    }
    \For{$g = 1$ to $T_{\text{gen}}$}{
        \ForEach{candidate $s_i \in \mathcal{P}_g$}{
            $\text{Fitness}(s_i) = \alpha \cdot \Delta_{\text{conf}}(s_i) + \beta \cdot \text{PR}(s_i)$\;
        }
        Rank candidates by fitness scores\;
        Select top $N_e$ elite candidates to form $\mathcal{P}^{\text{elite}}_g$\;
        
        \ForEach{pair $(s_1, s_2) \in \mathcal{P}^{\text{elite}}_g$ with probability $p_{\text{crossover}}$}{
            Create $s_{\text{new}}$ by combining connections and feature strategies from $s_1$ and $s_2$\;
            Add $s_{\text{new}}$ to $\mathcal{P}_{g+1}$\;
        }
        
        \ForEach{candidate $s_i \in \mathcal{P}_{g+1}$ with probability $p_{\text{mutate}}$}{
            Randomly modify edge connections and/or feature assignment\;
        }
    }
    
    Select best candidate $s^*$ from final generation $\mathcal{P}_{T_{\text{gen}}}$\;
    Inject nodes and edges from $s^*$ into graph; 
}
\Return $G'$
\end{algorithm}

As outlined in Algorithm~\ref{alg:graphtextack}, \method{} operates iteratively over multiple generations. It first  initializes a population of candidate injection strategies and explores the combinatorial search space through iterative selection, crossover, mutation, and fitness evaluation. It identifies an optimal injection strategy via this adaptive evolutionary process. 

The core challenge of selecting effective edge connections and feature representations is a discrete, high-dimensional combinatorial problem. Gradient-based optimization methods struggle in this setting, particularly under a black-box threat model where surrogate gradients may poorly approximate the target model’s decision boundaries. Exhaustive search is infeasible due to the exponential number of perturbation possibilities. Furthermore, each injection alters the graph topology, influencing downstream classification and necessitating adaptive re-optimization. 

To address these challenges, \method{} maintains a population of candidate solutions, each encoding a potential injection strategy, and iteratively refines them through evolutionary operators guided by our novel multi-objective fitness function.  
Our attack has four key components.

\noindent\textbf{(C1) Candidate Representation and Initialization.} 
We represent the population at generation $t$ as a set of $N_p$ candidate node injections $\mathcal{P}_t = \{ s_i \}_{i=1}^{N_p}$, where each candidate $s_i$ encodes: (1) a set of injected nodes $V'_i$ and their connections to existing nodes $E'_i \subseteq V' \times V$; and (2) a feature generation strategy specifying how node features $T'(v')$ are sampled for each injected node $v' \in V'_i$.

We randomly sample the initial population $\mathcal{P}_0$ by selecting edge connections and feature strategies according to uniform distributions, ensuring diverse exploration of the search space. 
The number of edges assigned to each injected node (budget) is sampled from the empirical degree distribution of nodes in $G$, ensuring that injected nodes exhibit realistic connectivity patterns and remain less detectable.

\noindent \textbf{(C2) Semantic Feature Generation.} 
Each injected node must also be assigned a textual description consistent with the graph's semantic space. Rather than directly optimizing feature embeddings, which would involve a complex, continuous search space (and risk generating unrealistic node attributes that do not fit in the feature space), we adopt a class-conditioned sampling strategy.

Each candidate $s_i$ specifies a class label $c \in \mathcal{Y}$ for its injected nodes. Features are then sampled from the empirical distribution of nodes belonging to class $c$. Formally, for each injected node $v'$, its feature embedding is sampled as 
\(X'(v') \sim p(X(v) \mid Y(v) = c)\), 
where $p(\cdot)$ denotes the empirical distribution estimated from labeled nodes $V_L$, which can be constructed by querying the target model to obtain pseudo labels. 
The class label assignment is part of the candidate's representation and is subject to evolutionary optimization during the search process.

Unlike prior attacks that perturb the structure or inject nodes with handcrafted features~\cite{wang2020scalable_AFGSM,zou2021_tdgia}, \method{} introduces class-conditioned semantic feature generation, this enables injected nodes to blend into the graph’s textual feature space in a way that preserves semantic consistency and improves imperceptibility, preventing injected nodes from appearing anomalous.

\noindent \textbf{(C3) Multi-modal Fitness Evaluation.} 
Since the effects of poisoning attacks cannot be directly measured, we design the fitness function as a multi-modal multi-objective optimization problem, balancing semantic-level disruption (measured by local prediction shifts) and structural-level influence (measured by global graph centrality).  

Formally, fitness is defined: \(\text{Fitness}(s_i) = \alpha \cdot \Delta_{\text{conf}}(s_i) + \beta \cdot \text{PR}(s_i)\), where $\Delta_{\text{conf}}(s_i)$ measures the prediction shift induced by an injected node, $\text{PR}(s_i)$ measures its structural influence, and $\alpha, \beta \geq 0$ balance the two components. 

 $\bullet$ \textit{Local prediction shift ($\Delta_{\text{conf}}$).}     
For target model $f$ that outputs SoftMax probabilities over class labels, let $C_v = \max(f_v(G))$ denote node's $v \in V$ maximum confidence score. After injecting node \(v'\) according to $s_i$, the new graph is $G_i'$ and the updated confidence is $C_v' = \max(f_v(G_i'))$. 
Thus, 
we define the local prediction shift as:
\(\Delta_{\text{conf}}(s_i) = \frac{1}{|\mathcal{N}_2(v')|} \sum_{v \in \mathcal{N}_2(v')} \left| C_v - C_v' \right|\),
where $\mathcal{N}_2(v')$ is the two-hop neighborhood of node $v'$. 

$\bullet$ \textit{Global, PageRank influence ($\text{PR}$).}  To capture the global influence of the injected node, we compute its PageRank score after injection into graph $G_i'$. The influence term is then defined as: \(\text{PR}(s_i) = \text{PR}(v')\), where \(v'\) is the node injected according to candidate \(s_i\). Nodes with higher centrality scores are likely to affect more nodes via message passing, so high PageRank indicates greater potential influence. 

This multi-modal fitness function ensures that the attack effectiveness is evaluated across both feature-based and topology-based vulnerabilities. Moreover, we combine prediction shift and PageRank influence to jointly optimize both localized adversarial impact and global structural importance. This ensures that the evolutionary search process favors candidates that not only perturb model predictions in their local neighborhoods but are also central in the graph.

Unlike prior adversarial attacks on graphs that primarily focus on structural modifications or localized prediction changes in isolation, \method{} jointly optimizes semantic perturbations and structural influence. 

\noindent \textbf{(C4) Multi-modal Evolutionary Operations.}  
At each generation, we apply the following evolutionary operations to refine the candidate population.

\textbf{Selection}: Candidate injection strategies are ranked based on fitness scores, and the top $N_e$ elite individuals $\mathcal{P}^{\text{elite}}_t \subseteq \mathcal{P}_t$ are retained to form the basis of the next generation.

\textbf{Crossover}: We select pairs of elite candidate injection strategies, and combine them to produce new candidates. 
Given parent candidates $s_1$ and $s_2$, a new candidate $s_{\text{new}}$ is created by splitting and recombining these components. 
Specifically, for the edge connections, a random crossover point $j$ is selected, and $s_{\text{new}}$ inherits the first $j$ connections from $s_1$ and the remaining connections from $s_2$:
\(E'_{\text{new}} = E'_{s_1}[:j] \Vert E'_{s_2}[j:]\). And the class sampling strategy is randomly chosen from one of the two parents. 
This \textit{multi-modal crossover} enables the evolutionary process to explore interactions between graph structure and node semantics, which is critical for effective attacks on LLM-enhanced GNNs.

\textbf{Mutation}: To maintain population diversity, each candidate is mutated with probability $p_{\text{mut}}$ by randomly altering edge connections or feature assignments.

\subsection{Search Space and Computational Complexity for Multi-modal Injection}
\label{sec:complexity}

\begin{proposition}[Search Space of Multi-Modal Injection]
Optimizing multi-modal node injection attacks in $G = (V, E)$ exactly requires searching over a space of size
$O\left(|V|^{r \cdot d_{\max}} \times |\mathcal{F}|^r\right)$,
which is exponentially large in the number of injected nodes \(r\), the maximum degree per injected node $d_{\max}$, and the feature set size \(|\mathcal{F}|\). 
\end{proposition}

The large search space motivates the evolutionary optimization design of \method, which maintains population diversity and uses mutation and crossover to explore the search space efficiently without exhaustive enumeration. 

\begin{lemma}[Polynomial-Time Evolutionary Approximation]
\label{lem:complexity}
Under fixed population size \(N_p\) and generation budget \(T_{\text{gen}}\), the runtime complexity of \textnormal{\textsc{GraphTextack}} per injection step is 
\(O\left(N_p \cdot T_{\text{gen}} \cdot (r \cdot d_{\max}^2 + |E|)\right)\),
where $r$ is the number of injected nodes, $d_{\max}$ is maximum degree per injected node, and $|E|$ is the original number of edges.
\end{lemma}

In practice, $d_{\max}$ is chosen to be small to ensure stealthy attacks. Since $N_p$ and $T_{\text{gen}}$ are user-controlled hyperparameters, the overall runtime of \method{} scales effectively linearly in the graph size \((|V|, |E|)\), and so it offers a scalable and tunable trade-off between search space and runtime. We empirically validate the efficiency of our attack in the next section. {In addition to the above complexity analysis, we formally characterize the synergy between structure and feature perturbations in multi-modal attacks, provide theoretical analysis of the advantages of multi-modal node injections, and analyze how local prediction shifts arise with node injection perturbations, in App. \ref{app:proofs}. }

\section{Experiments}
\label{sec:experiments}
We now present an extensive empirical evaluation of GraphTextack. Our analysis aims to answer the following key research questions:   
\textbf{(RQ1)}~Can our \method{} framework effectively degrade the node classification performance of LLM-enhanced GNN models through node injection?  How does our surrogate-free framework  in a black-box setting compare to strong baselines? 
\textbf{(RQ2)}~How does \method{} compare to baselines in terms of efficiency? 
\textbf{(RQ3)}~How do different components of our  framework contribute to the overall effectiveness?

\subsection{Empirical Setup} 
\label{sec: Empirical_setup}

\textbf{Datasets.} We conduct experiments on five widely used text-attributed graph benchmark datasets for node classification: Cora\cite{mccallum2000automating}, PubMed \cite{sen2008collective_cora_pubmed}, WikiCS \cite{mernyei2022wikiCS}, ogbn-arxiv\cite{hu2020open}, and ogbn-products \cite{hu2020open}.  
For ogbn-products, we use the sampled subset from~\cite{he2023harnessing}. For all datasets, we use the default splits, and obtain the raw textual attributes from \cite{liu2023one} and \cite{he2023harnessing}. We summarize the dataset statistics in App. \ref{app:results}. 

\textbf{Target Models.} We evaluate two representative paradigms of LLM-enhanced GNNs: 
(1)~\textbf{representation-level enhancer}: One-for-all \cite{liu2023one} is a universal graph representation learning framework that incorporates LLM-based text encodings to create a task-agnostic graph representation; 
(2)~\textbf{node-level enhancer}: In this setup, we directly perform node feature enhancement by encoding node attributes using the pretrained e5-large-v2 model~\cite{wang2024text}, and pass the resulting embeddings into a GCN model~\cite{luoclassic}. 

\textbf{Attack Baselines \& Hyperparameter Tuning.}
Since black-box node injection attack is an emerging area, there are few existing methods directly applicable for comparison. We compare our attack against 12 representative global attack methods and adaptations of general graph attacks. Specifically, we consider three textual attacks: (1)~BERT-Attack~\cite{Li20_BertAttack}, (2)~HotFlip~\cite{Ebrahimi18_HotFlip}, and (3)~VIPER~\cite{Eger19_VIPER}; 
(4) a gradient-based structural modification attack, PRBCD~\cite{geisler21_PRBCD}; 
(5) a Preferential attack~\cite{barabasi1999emergence}, a heuristics-based approach; 
popular attacks that we adapted to our node injection setup: 
(6) Nettack~\cite{zugner2018adversarial}: widely used targeted modification attack;  
(7) PSO~\cite{Zang_2020_pso}, a population-based targeted attack;
(8) TDGIA~\cite{zou2021_tdgia}, a targeted evasion attack;
(9) AFGSM~\cite{wang2020scalable_AFGSM}, a targeted poisoning attack;  
(10) $G^2A2C$~\cite{ju2023let_g2a2c}, an evasion injection attack based on direct model queries;  
(11) GANI~\cite{fang2024gani}, a global poisoning injection attack; and
(12) WTGIA~\cite{lei2024intruding_WTGIA}, a text-level node injection attack.  

{For our approach, we tune the population size, crossover and mutation probabilities, and the coefficients $\alpha$ and $\beta$ (which control the balance between local prediction disruption and global influence) through hyperparameter search. }
{We give more details in App. ~\ref{app:unimodal-results} and ~\ref{app:baselines}. }

\textbf{Evaluation Metrics.} 
For different node injection budgets, we report the average node classification accuracy after attack and stdev over five runs and attack generation runtime (wall-clock time). More effective methods lead to bigger  accuracy drop relative to the clean graph. 

\subsection{(\textbf{RQ1}) Overall Attack Performance} 
\label{sec:overall_performance}
First, we seek to assess the overall performance of \method{} by evaluating the classification accuracy drop across multiple datasets and target models. 
We vary the node injection budget $r$ in {$r$ = \{0.01, 0.03, 0.05\}}, intentionally keeping it low to ensure imperceptibility and avoid distribution shifts that would make the attack more detectable by defense mechanisms.
We give the post-attack classification accuracy results for the representation-level enhancer in Table \ref{tab:attack_results_ofa}, and for the feature-level enhancer in App.~\ref{app:results}. Among the three textual attacks, which modify \emph{all} the textual features, we provide the best-performing result per dataset ({`Best text*'}); since the budget does not apply in this case,  we report a single number. In the results table, {the best performance (i.e., the lowest accuracy) is bolded and the second best underlined.}
We discuss our key observations below. 

\setlength{\tabcolsep}{1mm}
\begin{table}[!htbp]
\centering
{\fontsize{9}{10}\selectfont
\begin{tabular}{llcccccc}
\toprule
\multirow{2}{*}{\makecell{\textbf{Dataset}\\(Clean Acc)}}
& \multirow{2}{*}{\textbf{Methods}}
& \multicolumn{3}{c}{\textbf{Classification Accuracy}} \\ \cmidrule{3-5}
& &  \textbf{r=0.01} & \textbf{r=0.03} & \textbf{r=0.05} \\
\midrule
\multirow{11}{*}{\makecell{\textbf{Cora} \\ (80.95 \\ {±0.87)}}} & Pref.
&77.34{±0.58}
&73.49{±0.52}
&72.39{±0.96}\\
& Best text* &  --------------- & 74.36{±1.68} & --------------- \\
& PRBCD
&79.33{±2.08}
&77.37{±1.03}
&74.75{±0.54}\\
& Nettack 
&76.83{±0.27}
&68.47{±0.20}
&\underline{62.59{±0.30}}\\
& PSO 
&77.83{±0.60}
&73.14{±0.52}
&68.07{±0.39}\\
& TDGIA 
&75.79{±0.12}
&71.26{±0.26}
&66.26{±0.49}\\
& AFGSM 
&\underline{75.47{±0.37}}
&69.75{±0.75}
&63.89{±0.22}\\
& $G^{2}A2C$ 
&77.20{±0.82}
&\underline{67.79{±0.67}}
&63.74{±0.69}\\
& GANI 
&77.18{±0.77}
&70.98{±0.96}
&66.06{±0.70}\\
& WTGIA 
&76.35{±0.38}
&69.63{±0.43}
&65.80{±0.36}\\
& Ours &\textbf{73.99}{±0.78}
&\textbf{65.75}{±0.81}
&\textbf{62.02}{±0.64}\\
\midrule
\multirow{11}{*}{\makecell{\textbf{PubMed}\\(71.65 \\ ±0.77)}}
& Pref.
&67.48{±0.57}
&59.34{±0.60}
&50.97{±0.85}\\

& Best text*  & --------------- & {71.78{±0.69}} & ---------------\\
& PRBCD 
&71.13{±0.18}
&68.60{±1.03}
&65.04{±0.90}\\

& Nettack
&64.28{±0.53}
&51.00{±0.45}
&45.14{±0.27}\\
& PSO
&65.62{±0.35}
&49.61{±0.38}
&43.58{±0.71}\\
& TDGIA
&64.70{±0.56}
&51.33{±0.39}
&43.37{±0.43}\\
& AFGSM
&62.96{±0.59}
&53.98{±0.77}
&44.11{±0.24}\\
& $G^{2}A2C$
&\underline{62.10{±0.41}}
&54.96{±0.24}
&43.78{±0.76}\\
& GANI
& 64.33{±0.74}
& 54.59{±1.23}
& 45.61{±0.48}\\
& WTGIA
&65.89{±0.96}
&\underline{48.96{±0.46}}
&\underline{43.10{±0.74}}\\
& Ours
&\textbf{60.78}{±0.25}
&\textbf{48.43}{±0.63}
&\textbf{42.05}{±0.68}\\
\midrule
\multirow{11}{*}{\makecell{\textbf{WikiCS}\\(76.31\\±0.94)}}
& Pref.
& 73.33{±1.66}
& 72.02{±0.54}
& 69.95{±0.61}\\

& Best text*  & ---------------  & {74.34{±0.21}} & --------------- \\

& PRBCD 
&72.96{±1.04}
&70.13{±0.52}
&67.38{±0.51}\\
& Nettack
&72.92{±0.70}
&65.29{±0.92}
&62.68{±0.67}\\
& PSO
&73.89{±1.67}
&67.39{±0.64}
&64.77{±1.00}\\
& TDGIA
&72.33{±0.65}
&65.67{±0.83}
&62.23{±0.78}\\
& AFGSM
&72.47{±0.71}
&67.25{±0.27}
&61.71{±1.16}\\
& $G^{2}A2C$
&72.90{±0.36}
&65.82{±0.50}
&62.46{±0.71}\\
& GANI
& 73.22{±0.86}
& 69.03{±0.89}
& 63.83{±0.55}\\
& WTGIA
&\underline{71.87{±1.17}}
&\textbf{64.17}{±0.84}
&\textbf{60.95}{±0.26}\\
& Ours
& \textbf{71.69}{±1.04}
& \underline{64.58{±0.69}}
& \underline{61.35{±1.33}}\\

\midrule
\multirow{11}{*}{\makecell{\textbf{ogbn-} \\ \textbf{arxiv}\\(75.44\\±1.09)}}
& Pref.
&74.86{±0.44}
&72.53{±0.37}
&69.27{±0.75}\\

& Best text*  & --------------- & {72.13{±0.65}} & --------------- \\

& PRBCD
&74.98{±0.49}
&73.38{±0.70}
&72.38{±0.79}\\
& Nettack
&73.96{±0.89}
&69.93{±0.33}
&\underline{66.76{±0.55}}\\
& PSO
&\underline{72.03{±0.46}}
&69.82{±1.14}
&67.87{±0.89}\\
& TDGIA
&73.63{±0.47}
&70.05{±0.32}
&66.97{±0.21}\\
& AFGSM
&73.09{±0.90}
&71.24{±0.20}
&68.82{±0.45}\\
& $G^{2}A2C$
&72.63{±0.34}
&\underline{69.50{±0.68}}
&66.67{±0.24}\\
& GANI
&73.61{±0.91}
&70.31{±0.42}
&68.17{±0.36}\\
& WTGIA
&72.38{±0.72}
&70.04{±0.76}
&67.61{±0.98}\\
& Ours
&\textbf{71.95}{±0.38}
&\textbf{68.23}{±0.85}
&\textbf{66.61}{±0.78}\\

\midrule
\multirow{11}{*}{\makecell{\textbf{ogbn-} \\ \textbf{products}\\(83.51\\±1.05)}}
& Pref.
&80.32{±1.46}
&77.92{±0.56}
&75.29{±0.95}\\

& Best text*  & ---------------& {75.45{±0.23}} & --------------- \\

& PRBCD
&81.09{±0.27}
&79.28{±0.29}
&77.85{±0.14}\\
& Nettack
&80.09{±0.87}
&76.55{±0.44}
&70.23{±0.29}\\
& PSO
&\textbf{78.91}{±0.95}
&75.16{±0.45}
&70.71{±0.40}\\
& TDGIA
&80.22{±0.95}
&75.45{±0.99}
&72.11{±0.93}\\
& AFGSM
&79.76{±0.68}
&75.05{±0.51}
&70.33{±0.94}\\
& $G^{2}A2C$
&79.05{±0.23}
&74.90{±0.17}
&71.61{±0.32}\\
& GANI
&\underline{78.93{±0.37}}
&\underline{74.83{±0.99}}
&\underline{69.51{±0.64}}\\
& WTGIA
&79.37{±0.54}
&75.55{±0.28}
&71.83{±0.21}\\
& Ours
&78.99{±0.63}
&\textbf{74.26}{±0.64}
&\textbf{69.05}{±0.51}\\ 

\bottomrule
\end{tabular}%
}
\caption{[Representation-level enhancer model] Classification accuracy after injection attacks. }
\label{tab:attack_results_ofa}
\end{table}

\textit{Observations.}
\method{} generally achieves more significant degradation in classification accuracy compared to baseline attacks, across datasets and injection budgets, demonstrating strong generalization and scalability. This is achieved through the multi-modal, multi-objective fitness evaluation and evolutionary operations. Jointly optimizing for prediction impact and structural centrality allows the attack to more effectively explore the search space, leading to more effective injections. Notably, even the strongest textual attacks, despite modifying all node features, still underperform compared to \method, highlighting the advantage of multi-modal attacks. 

\method{} directly queries the target model, avoiding reliance on surrogate model approximations. This allows it to better adapt to the decision boundaries of the target model, regardless of architectural differences (representation-level in Table~\ref{tab:attack_results_ofa} vs.\ feature-level enhancer in Table \ref{tab:attack_results_feat}). That makes it more effective in settings with a significant mismatch between the surrogate and target models, such as the representation-level enhancer.

\subsubsection{(\textbf{RQ2}) Runtime}
\label{sec:runtime_analysis}

\begin{figure}[!htb]
  \centering
  \includegraphics[width=1\linewidth]{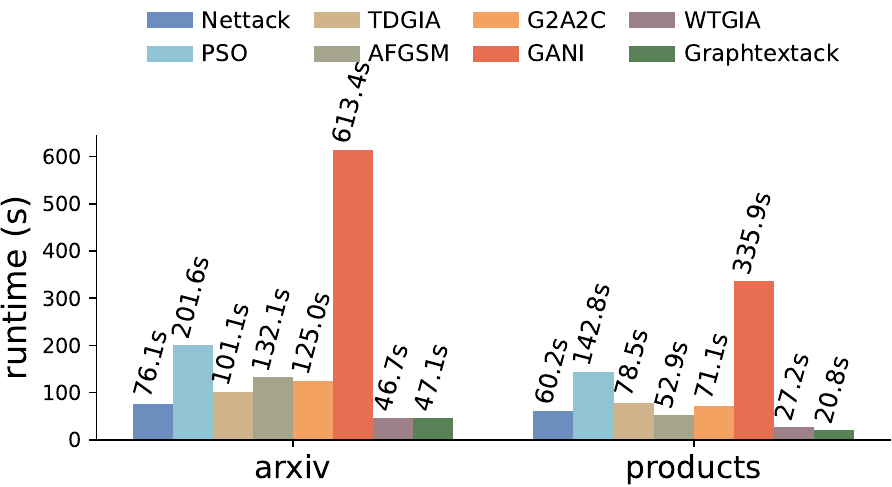}
  \caption{Runtime per injection on the representation-level enhancer target model.}
  \label{fig:runtime-large}
\end{figure}

In addition to attack effectiveness, computational efficiency is another important practical consideration for node injection attacks. Figure \ref{fig:runtime-large} compares the time required to generate each injection attack for the baseline methods and our proposed \method{} across the largest datasets. We provide results for all datasets in App.~\ref{app:results}.

\textit{Observations.} \method{} achieves the lowest average runtime, showing clear efficiency improvement over baselines. 
By avoiding gradient-based optimization and surrogate modeling, it eliminates major computational bottlenecks, enabling more scalable attacks.

\subsubsection{(\textbf{RQ3}) Ablation study}
\label{sec:ablation}

To verify the effectiveness of our approach, we analyze the influence of each main component of the optimization framework through an ablation study. We systematically remove crossover, mutation, and both components of the fitness function. For each ablation setting, we evaluate the post-attack classification accuracy under different injection budgets and compare it to the full model.

\begin{table}[t!]
    \centering
    {\fontsize{9pt}{10pt}\selectfont
    \begin{tabular}{lcc}
        \toprule
        \textbf{Variants} & \textbf{r=0.01} &  \textbf{r=0.05} \\ 
        \midrule
         \method 
         & 71.69{±1.04}  
         & 61.35{±1.33} \\
         Without crossover 
         & 73.50{±0.71}  
         & 64.42{±0.59} \\
         Without mutation 
         & 73.39{±0.85} 
         & 65.19{±0.41} \\
         Without pred. shift& 73.77{±0.67} 
         & 66.83{±0.54} \\
         Without PageRank& 71.92{±0.78}  
         & 61.97{±1.00} \\
        \bottomrule
    \end{tabular}
    }
    \caption{Ablation study for \method{} on WikiCS: Accuracy after injection on the representation-level enhancer model.}
    \label{tab:ablation}
\end{table}

\textit{Observations.}  Table \ref{tab:ablation} shows that removing any of the main mechanisms of the evolutionary optimization process noticeably reduces attack effectiveness. 
This highlights the importance of population diversity and search space exploration in the optimization process. We observe consistent trends across different datasets and budgets (App.~\ref{app:results}). 

\section{Conclusion}
We introduced \method{}, a novel black-box node injection attack against LLM-enhanced GNNs. Unlike existing approaches that rely on surrogate models, \method{} directly queries the target model during the optimization process, leading to a more effective and efficient strategy. 
We demonstrated that our method outperforms 12 baselines, including heuristics-based methods and state-of-the-art evolutionary optimization models. 
We also showed that utilizing direct model queries is more computationally efficient than surrogate model-based approaches. 
Our work provides a foundation for future work to explore defense mechanisms tailored for LLM-enhanced GNNs{, such as joint adversarial training and data filtering. }

\noindent \textbf{Limitations.} 
{\method{} assumes the attacker can inject nodes with controllable structure and text, which may not be feasible in some applications or domains.} 
Additionally, our current approach follows a simple class-conditioned feature sampling strategy, which, while effective, may be further improved by learning more expressive or task-adaptive feature generation strategies. 

\section*{Ethical Statement}
{Our work studies attacks against LLM-enhanced GNNs. All experiments were performed on public datasets and models. 
While our methods expose potential vulnerabilities, the goal is to highlight real risks in order to promote the development of more robust and secure models, as well as defense mechanisms.}

\section*{Acknowledgments}

 This material in this work is supported by the National Science Foundation under IIS~2504090, IIS 2212143 and CAREER Grant No.IIS 1845491. Any opinions, findings, and conclusions or recommendations expressed in this material are those of the author(s) and do not necessarily reflect the views of the National Science
 Foundation or other funding parties. PT completed this work prior to joining Amazon.

{\small
\bibliography{aaai2026}
}

\clearpage
\appendix
\section{Motivation for Multi-modal Attacks: Detailed Evaluation of Robustness of LLM-enhanced GNNs to Uni-modal Attacks}
\label{app:unimodal-results}

Recent studies have begun to explore adversarial vulnerabilities of LLM-enhanced GNNs under uni-modal attacks ~\cite{guo2024learning}. For completeness and alignment with our main experiments, we evaluate several representative \textit{uni-modal} attacks to assess how much performance degradation they can achieve. 

\noindent $\bullet$ \textbf{Textual perturbation attacks}: This type of attack introduces small, imperceptible modifications to the textual attributes associated with nodes, aiming to manipulate the LLM-generated features and thus influencing GNN output. We use the OpenAttack library \cite{Zeng20_openattack} to implement three perturbations at different granularity levels: BERT-ATTACK, HotFlip, and VIPER.

\noindent $\bullet$ \textbf{Structural perturbation attacks}: The other common approach is modifying the graph structure by randomly adding or removing edges. These perturbations can disrupt the neighborhood aggregation process in GNNs. A baseline involves randomly adding heterophilous edges or removing homophilous edges~\cite{Zhu22_HeterophilyRobustness}. We also evaluate Projected Randomized Block Coordinate Descent (PRBCD) \cite{geisler21_PRBCD}, a gradient-based structural attack that optimizes edge perturbations strategically to maximize classification performance drop.

We use the same node classification accuracy metric defined in Sec \ref{sec: Empirical_setup}. For each attack, we report post-attack accuracy under the poisoning setting, averaged over five runs.

\textbf{Results.} Table \ref{tab:naive_poisoning_results} summarizes the classification accuracy after applying various attacks.
While some methods degrade performance, the overall impact is limited. Textual perturbations over the entire training set result in only marginal accuracy drops. Structural attacks are more effective but typically require large-scale modifications (e.g., impacting 20-25\% of the edges).

This study highlights the fundamental limitations of these graph modification attacks. Moreover, textual perturbations require direct control over the attributes of existing data, and structural attacks require control over the connectivity of existing data, which is often unrealistic and easy to detect. They do not fully capture a realistic threat model, as attackers typically cannot directly control the existing data in real-world scenarios.

\setlength{\tabcolsep}{1mm}
\renewcommand{\arraystretch}{0.85}
\setlength{\extrarowheight}{-1pt}
\begin{table*}[!htb]
    \centering
    \scriptsize
    
    \begin{tabular}{ll llllllllll}

    \toprule
 & & \multicolumn{2}{c}{\textbf{Cora}}
 & \multicolumn{2}{c}{\textbf{PubMed}} 
 & \multicolumn{2}{c}{\textbf{WikiCS}} 
 & \multicolumn{2}{c}{\textbf{ogbn-arxiv}}
 & \multicolumn{2}{c}{\textbf{ogbn-products}}\\
    \cmidrule(lr){3-4} \cmidrule(lr){5-6} \cmidrule(lr){7-8} \cmidrule(lr){9-10} \cmidrule(lr){11-12}
    & \textbf{Attack} 
    & \makecell{\textbf{Repres.-} \\ \textbf{level} \\ \textbf{enhancer}}
    & \makecell{\textbf{Feature-} \\ \textbf{level} \\ \textbf{enhancer}}
    & \makecell{\textbf{Repres.-} \\ \textbf{level} \\ \textbf{enhancer}}
    & \makecell{\textbf{Feature-} \\ \textbf{level} \\ \textbf{enhancer}}  
    & \makecell{\textbf{Repres.-} \\ \textbf{level} \\ \textbf{enhancer}}
    & \makecell{\textbf{Feature-} \\ \textbf{level} \\ \textbf{enhancer}}
    & \makecell{\textbf{Repres.-} \\ \textbf{level} \\ \textbf{enhancer}}
    & \makecell{\textbf{Feature-} \\ \textbf{level} \\ \textbf{enhancer}}   
    & \makecell{\textbf{Repres.-} \\ \textbf{level} \\ \textbf{enhancer}}
    & \makecell{\textbf{Feature-} \\ \textbf{level} \\ \textbf{enhancer}}
    \\  
    \midrule
    & \makecell{No attack \\ (Clean)}
&80.95 ± {\scriptsize 0.87} 
&79.49 ± {\scriptsize 1.04}
&76.23 ± {\scriptsize 0.86} 
&74.38 ± {\scriptsize 1.64}
&76.31 ± {\scriptsize 0.94}
&76.36 ± {\scriptsize 0.82}
&75.44 ± {\scriptsize 1.09}
&69.53 ± {\scriptsize 0.36}
&83.51 ± {\scriptsize 1.05}
&82.86 ± {\scriptsize 0.50}
 \\
    \midrule
    \multirow{3}{*}{\makecell{\textbf{Textual}\\\textbf{Perturbations}}} & BertAttack 
&78.32 ± {\scriptsize 1.22}
&76.21 ± {\scriptsize 0.85}
&72.65 ± {\scriptsize 0.77}
&73.70 ± {\scriptsize 1.05}
&76.13 ± {\scriptsize 0.55}
&75.76 ± {\scriptsize 0.76}
&73.66 ± {\scriptsize 0.38}
&68.72 ± {\scriptsize 0.42}
&78.35 ± {\scriptsize 0.33}
&80.45 ± {\scriptsize 0.22}
\\
    & HotFlip 
&74.99 ± {\scriptsize 0.96}
&73.37 ± {\scriptsize 0.91}
&73.83 ± {\scriptsize 0.94}
&74.19 ± {\scriptsize 0.92}
&74.34 ± {\scriptsize 0.21}
&76.46 ± {\scriptsize 0.33}
&74.19 ± {\scriptsize 0.46}
&68.39 ± {\scriptsize 0.16}
&78.20 ± {\scriptsize 0.45}
&79.86 ± {\scriptsize 0.41}
\\
    & Viper 
&74.36 ± {\scriptsize 1.68} 
&75.64 ± {\scriptsize 0.71}
&71.78 ± {\scriptsize 0.69}
&73.06 ± {\scriptsize 0.70}
&75.88 ± {\scriptsize 0.46}
&74.79 ± {\scriptsize 0.44}
&72.13 ± {\scriptsize 0.65}
&67.33 ± {\scriptsize 0.75}
&75.45 ± {\scriptsize 0.23}
&79.44 ± {\scriptsize 0.36}

 \\
    \midrule 
    \multirow{15}{*}{\rotatebox[origin=c]{90}{\textbf{Structural Perturbations}}} 

\multirow{5}{*}{\makebox[2cm][c]{Remove Edges}}
    
    & 5\%
&79.62 ± {\scriptsize 0.52} 
&79.32 ± {\scriptsize 0.36}
&68.66 ± {\scriptsize 0.67}
&71.96 ± {\scriptsize 0.42}
&76.59 ± {\scriptsize 0.65}
&75.79 ± {\scriptsize 0.46}
&74.47 ± {\scriptsize 0.43}
&69.29 ± {\scriptsize 0.23}
&79.72 ± {\scriptsize 0.16}
&80.34 ± {\scriptsize 0.69}\\
    & 10\%
&77.97 ± {\scriptsize 0.35} 
&78.86 ± {\scriptsize 0.38}
&68.19 ± {\scriptsize 1.37}
&70.78 ± {\scriptsize 1.53}
&75.25 ± {\scriptsize 0.83}
&75.17 ± {\scriptsize 0.59}
&73.87 ± {\scriptsize 0.39}
&68.89 ± {\scriptsize 0.17}
&79.12 ± {\scriptsize 0.42}
&80.07 ± {\scriptsize 0.34}\\
    & 15\% 
&78.11 ± {\scriptsize 0.54} 
&77.94 ± {\scriptsize 0.90}
&67.29 ± {\scriptsize 0.62}
&70.31 ± {\scriptsize 1.04}
&74.41 ± {\scriptsize 0.33}
&74.16 ± {\scriptsize 0.41}
&73.04 ± {\scriptsize 0.58}
&68.58 ± {\scriptsize 0.35}
&77.76 ± {\scriptsize 0.59}
&79.56 ± {\scriptsize 0.14}\\
    & 20\%  
&78.73 ± {\scriptsize 0.54} 
&77.19 ± {\scriptsize 0.33}
&67.53 ± {\scriptsize 0.83}
&70.15 ± {\scriptsize 0.56}
&73.00 ± {\scriptsize 0.17}
&74.41 ± {\scriptsize 0.58}
&72.85 ± {\scriptsize 0.55}
&67.97 ± {\scriptsize 0.31}
&76.39 ± {\scriptsize 0.33}
&78.98 ± {\scriptsize 0.10}\\
    & 25\%  
&78.18 ± {\scriptsize 0.41} 
&76.38 ± {\scriptsize 0.60}
&66.82 ± {\scriptsize 0.72}
&68.18 ± {\scriptsize 0.54}
&73.88 ± {\scriptsize 0.38}
&74.09 ± {\scriptsize 0.17}
&72.40 ± {\scriptsize 0.66}
&67.46 ± {\scriptsize 0.29}
&75.40 ± {\scriptsize 0.54}
&78.62 ± {\scriptsize 0.43}\\ 

\cmidrule(lr){2-12}
\multirow{5}{*}{\makebox[2cm][c]{Add Edges}}

    & 5\% 
&77.90 ± {\scriptsize 0.30} 
&78.96 ± {\scriptsize 0.92}
&72.85 ± {\scriptsize 0.32}
&73.99 ± {\scriptsize 0.72}
&77.63 ± {\scriptsize 0.71}
&76.92 ± {\scriptsize 0.50}
&74.93 ± {\scriptsize 0.62}
&67.94 ± {\scriptsize 0.37}
&79.81 ± {\scriptsize 0.27}
&79.21 ± {\scriptsize 0.38}\\
    & 10\%  
&76.34 ± {\scriptsize 0.65} 
&75.64 ± {\scriptsize 0.68}
&72.71 ± {\scriptsize 0.54}
&73.06 ± {\scriptsize 0.39}
&76.92 ± {\scriptsize 0.59}
&77.50 ± {\scriptsize 0.28}
&73.38 ± {\scriptsize 0.56}
&67.48 ± {\scriptsize 0.81}
&78.98 ± {\scriptsize 0.18}
&78.85 ± {\scriptsize 0.47}\\
    & 15\%  
&75.18 ± {\scriptsize 0.28} 
&74.78 ± {\scriptsize 0.85}
&72.88 ± {\scriptsize 0.78}
&73.58 ± {\scriptsize 0.89}
&75.80 ± {\scriptsize 0.37}
&75.95 ± {\scriptsize 0.63}
&73.26 ± {\scriptsize 0.32}
&67.08 ± {\scriptsize 0.21}
&78.28 ± {\scriptsize 0.56}
&78.40 ± {\scriptsize 0.51}\\
    & 20\%  
&73.66 ± {\scriptsize 0.44} 
&73.95 ± {\scriptsize 0.40}
&72.17 ± {\scriptsize 0.47}
&73.34 ± {\scriptsize 1.14}
&74.37 ± {\scriptsize 0.29}
&75.13 ± {\scriptsize 0.35}
&71.68 ± {\scriptsize 0.99}
&66.77 ± {\scriptsize 0.32}
&77.69 ± {\scriptsize 0.49}
&77.64 ± {\scriptsize 0.32}\\
    & 25\%  
&72.85 ± {\scriptsize 0.44} 
&73.18 ± {\scriptsize 0.39}
&70.89 ± {\scriptsize 1.20}
&70.24 ± {\scriptsize 0.32}
&71.66 ± {\scriptsize 0.17}
&72.21 ± {\scriptsize 0.78}
&71.56 ± {\scriptsize 0.41}
&66.55 ± {\scriptsize 0.28}
&76.02 ± {\scriptsize 0.74}
&76.71 ± {\scriptsize 0.30}\\ 

\cmidrule(lr){2-12}
\multirow{5}{*}{\makebox[2cm][c]{PRBCD}}

    & 1\% 
&79.33 ± {\scriptsize 2.08}
&78.55 ± {\scriptsize 0.89}
&71.13 ± {\scriptsize 0.18}
&70.74 ± {\scriptsize 0.47}
&72.96 ± {\scriptsize 1.04}
&73.79 ± {\scriptsize 0.83}
&74.98 ± {\scriptsize 0.49}
&68.29 ± {\scriptsize 0.26}
&81.09 ± {\scriptsize 0.27}
&80.17 ± {\scriptsize 0.18}\\
    & 2\% 
&78.68 ± {\scriptsize 1.25} 
&76.68 ± {\scriptsize 0.23}
&69.34 ± {\scriptsize 0.56}
&68.91 ± {\scriptsize 0.86}
&71.75 ± {\scriptsize 0.82}
&71.42 ± {\scriptsize 0.91}
&73.81 ± {\scriptsize 0.38}
&67.92 ± {\scriptsize 0.22}
&80.02 ± {\scriptsize 0.39}
&78.82 ± {\scriptsize 0.46}\\
    & 3\% 
&77.37 ± {\scriptsize 1.03} 
&75.22 ± {\scriptsize 0.38}
&68.60 ± {\scriptsize 0.47}
&68.04 ± {\scriptsize 0.81}
&70.13 ± {\scriptsize 0.52}
&69.37 ± {\scriptsize 0.46}
&73.38 ± {\scriptsize 0.70}
&67.54 ± {\scriptsize 0.16}
&79.28 ± {\scriptsize 0.29}
&78.02 ± {\scriptsize 0.35}\\
    & 4\% 
&75.61 ± {\scriptsize 1.46} 
&74.52 ± {\scriptsize 0.16}
&66.72 ± {\scriptsize 0.68}
&66.58 ± {\scriptsize 0.55}
&69.92 ± {\scriptsize 0.25}
&69.21 ± {\scriptsize 0.79}
&72.66 ± {\scriptsize 0.92}
&66.90 ± {\scriptsize 0.63}
&78.20 ± {\scriptsize 0.33}
&77.44 ± {\scriptsize 0.64}\\
    & 5\% 
&74.75 ± {\scriptsize 0.54} 
&73.62 ± {\scriptsize 0.36}
&65.04 ± {\scriptsize 0.90}
&65.22 ± {\scriptsize 0.68}
&67.38 ± {\scriptsize 0.51}
&67.92 ± {\scriptsize 0.42}
&72.38 ± {\scriptsize 0.79}
&66.04 ± {\scriptsize 0.27}
&77.85 ± {\scriptsize 0.14}
&75.46 ± {\scriptsize 0.27}\\

    \bottomrule
    \end{tabular}
    \caption{Classification accuracy after existing uni-modal poisoning attacks: textual and structural perturbations. Lower accuracy indicates stronger attack effectiveness.}
    \label{tab:naive_poisoning_results}
\end{table*}

\section{Additional Theoretical Analysis \& Proofs}
\label{app:proofs}

\subsection{Formal Motivation for Multi-Modal Node Injection}
\label{sec:theory-motivation}

Here, we take a global perspective and aim to characterize the multi-modal adversarial synergy. In App.~\ref{app:local}, we take a local perspective and characterize how local prediction shifts arise when an injected node perturbs both its feature-space alignment and its structural neighborhood.  Let $\Delta_{\text{acc}}(a)$ denote the degradation in model accuracy caused by attack $a$.
Let $\mathcal{S}$ denote the set of structure-only attacks, $\mathcal{F}$ the set of feature-only attacks, and $\mathcal{J}$ the set of joint structure-and-feature attacks. 

\begin{lemma}[Multi-Modal Adversarial Synergy]
\label{lemma:multimodal_synergy}
Suppose the target model $f$ is locally Lipschitz-continuous with respect to both structure and features, with constants $L_E, L_X > 0$, and exhibits a non-trivial interaction between structure and feature perturbations modeled by a cross-term coefficient $\gamma > 0$. 
Then for small perturbations
$\Delta_{\text{acc}}(a) \approx L_E \cdot d_E(a) + L_X \cdot d_X(a) + \gamma \cdot d_E(a) d_X(a),$ 
where $d_E(a)$ and $d_X(a)$ are the magnitudes of structure and feature perturbations respectively. 
Moreover:
\begin{equation}
\max_{a \in \mathcal{J}} \Delta_{\text{acc}}(a) > \max\left\{ \max_{s \in \mathcal{S}} \Delta_{\text{acc}}(s),\; \max_{f \in \mathcal{F}} \Delta_{\text{acc}}(f) \right\}.
\end{equation}
\end{lemma}

\begin{corollary}[Efficiency]
Under the setting of Lemma~\ref{lemma:multimodal_synergy}, multi-modal perturbations achieve greater degradation than uni-modal perturbations under any fixed total perturbation budget.
\end{corollary}

\subsection{Proof for the formal motivation}

\begin{proof}[Proof for Lemma~\ref{lemma:multimodal_synergy}]
By local Lipschitz continuity, for small perturbations $a$ affecting structure and features jointly, we can approximate the change in model outputs by:

\begin{equation}
\begin{split}
\| f(E', X') - f(E, X) \| \lesssim\,
    L_E \cdot d_E(a) + L_X \cdot d_X(a) \\
    + \gamma \cdot d_E(a) d_X(a)
\end{split}
\end{equation}

where 
$L_E \cdot d_E(a)$ models the first-order sensitivity to structural changes,
$L_X \cdot d_X(a)$ models the first-order sensitivity to feature changes, and
$\gamma \cdot d_E(a) d_X(a)$ captures second-order cross-modal amplification.

Structure-only attacks ($s \in \mathcal{S}$) incur a shift $L_E \cdot d_E(s)$.
Feature-only attacks ($f \in \mathcal{F}$) incur a shift $L_X \cdot d_X(f)$. 
However, joint structure-and-feature attacks ($a \in \mathcal{J}$) benefit from the positive cross-term $\gamma \cdot d_E(s) d_X(f) > 0$, resulting in strictly larger prediction shifts.

Thus, multi-modal perturbations achieve strictly greater degradation in model accuracy than either structure-only or feature-only perturbations alone.
\end{proof}

\subsection{Approximate Local Prediction Shift Bound}
\label{app:local}

We characterize how local prediction shifts arise when an injected node perturbs both its feature-space alignment and its structural neighborhood.

\begin{lemma}[Approximate Shift Bound]
\label{lemma:local}
Let $\mathcal{N}_2(v')$ denote the two-hop neighborhood of an injected node $v'$ in the perturbed graph $G'$.
Assume that the model $f$ computes node embeddings via neighborhood aggregation (e.g., mean aggregation) and that $f$ is locally Lipschitz-continuous\footnote{This Lipschitz continuity bound holds for standard message-passing GNNs with activation functions such as ReLU and SoftMax, as shown in \cite{Aladi18-Lipschitz}. } with respect to feature and structural perturbations.
Then, for any node $v \in \mathcal{N}_2(v')$, there exists a constant $L > 0$ such that:
\begin{equation}
\Delta_{\text{conf}}(v) \leq L \cdot \left( \| X(v') - \bar{X}_{N(v')} \| + \frac{|N(v')|}{\deg(v) + 1} \right)
\end{equation}
where $\bar{X}_{N(v')}$ denotes the mean feature vector of $v'$'s neighbors and $\deg(v)$ is the degree of node $v$.
\end{lemma}

\begin{proof}[Proof]

Let $G = (V, E, T, Y)$ be a text-attributed graph (TAG), where $\mathbf{X} \in \mathbb{R}^{|V| \times d}$ denotes node feature embeddings derived from textual attributes via an encoding function $\phi : \mathcal{T} \to \mathbb{R}^d$, i.e., $\mathbf{x}_v = \phi(T(v))$ for each $v \in V$. Let $\mathbf{A} \in \{0,1\}^{|V| \times |V|}$ denote the adjacency matrix of $G$, where $\mathbf{A}_{uv} = 1$ if $(u,v) \in E$ and $0$ otherwise.

Suppose a GNN model $f$ computes node embeddings $\mathbf{h}_v$ via message passing of the form:
\begin{equation}
\mathbf{h}_v = \text{AGG}\left( \{ \mathbf{x}_u : u \in \mathcal{N}(v) \} \right),
\end{equation}
where $\mathcal{N}(v)$ denotes the 1-hop neighbors of node $v$ and $\text{AGG}$ is a permutation-invariant aggregation function (e.g., mean aggregation). 
The model outputs SoftMax class probabilities $\mathbf{p}_v = \text{SoftMax}(\mathbf{W}_{\text{out}} \mathbf{h}_v)$, where $\mathbf{W}_{\text{out}}$ is the classifier weight matrix. 

Let $v' \notin V$ be an injected adversarial node connected to a subset $N(v') \subseteq V$ of existing nodes. Let $G'$ denote the graph after injection, with updated adjacency matrix $\mathbf{A}'$ and feature matrix $\mathbf{X}'$ (with $\mathbf{x}_{v'}$ appended).

We define the local confidence shift for a node $v \in \mathcal{N}_2(v')$ as:
\begin{equation}
\Delta_{\text{conf}}(v) := \left| \max\left( f(G')_v \right) - \max\left( f(G)_v \right) \right|.
\end{equation}

We now bound $\Delta_{\text{conf}}(v)$:  
Since $f$ is composed of aggregation followed by a classifier, and assuming that both the aggregation function $\text{AGG}$ and the classifier are Lipschitz continuous with constant $L > 0$, we have:
\begin{equation}
\Delta_{\text{conf}}(v) \leq L \cdot \| \mathbf{h}_v' - \mathbf{h}_v \|,
\end{equation}
where $\mathbf{h}_v'$ is the embedding of node $v$ after the injection (in $G'$).

The change $\| \mathbf{h}_v' - \mathbf{h}_v \|$ can be decomposed into two sources:

\begin{enumerate}
    \item \textbf{Feature Perturbation:} Adding $v'$ introduces a feature $\mathbf{x}_{v'}$ that may be different from the mean of $v$'s existing neighbors' features.
    
    \item \textbf{Structural Perturbation:} Adding $v'$ modifies the degree of $v$ from $\deg(v)$ to $\deg(v) + 1$, thus it alters the aggregation normalization.
\end{enumerate}
\paragraph{Step 1: Impact of Feature Perturbation.}
In a GNN with mean aggregation, the new aggregated feature for $v$ after injection becomes:
\begin{equation}
\bar{\mathbf{x}}_{\mathcal{N}'(v)} = \frac{1}{\deg(v)+1} \left( \sum_{u \in \mathcal{N}(v)} \mathbf{x}_u + \mathbf{x}_{v'} \right).
\end{equation}
Thus, the perturbation to the aggregated feature is:
\begin{equation}
\bar{\mathbf{x}}_{\mathcal{N}'(v)} - \bar{\mathbf{x}}_{\mathcal{N}(v)} = \frac{1}{\deg(v)+1} \left( \mathbf{x}_{v'} - \bar{\mathbf{x}}_{\mathcal{N}(v)} \right).
\end{equation}
Taking norms:
\begin{equation}
\| \bar{\mathbf{x}}_{\mathcal{N}'(v)} - \bar{\mathbf{x}}_{\mathcal{N}(v)} \| \leq \frac{1}{\deg(v)+1} \| \mathbf{x}_{v'} - \bar{\mathbf{x}}_{\mathcal{N}(v)} \|.
\label{eq:feat-bound}
\end{equation}

\paragraph{Step 2: Impact of Structural Perturbation (Degree Normalization).}
Adding $v'$ also perturbs the degree normalization of neighbors. 
Specifically, the relative effect on aggregation for a neighbor $v$ is proportional to the number of newly added connections $|N(v')|$, scaled by the new degree $\deg(v) + 1$. 
Thus, the additional perturbation magnitude due to structural reweighting is bounded by:
\begin{equation}
\frac{|N(v')|}{\deg(v)+1}.
\end{equation}
This term captures the increase in influence from the newly added neighbor(s).

\paragraph{Step 3: Combining Both Effects.}
By Lipschitz continuity of the aggregation and classification functions with constant $L > 0$, the total change in confidence satisfies:
\begin{equation}
\Delta_{\text{conf}}(v) \leq L \cdot \left( \| \bar{\mathbf{x}}_{\mathcal{N}'(v)} - \bar{\mathbf{x}}_{\mathcal{N}(v)} \| + \frac{|N(v')|}{\deg(v)+1} \right).
\end{equation}
Substituting the bound from Step 1 (Eq.~\eqref{eq:feat-bound}), we then obtain:
\begin{equation}
\Delta_{\text{conf}}(v) \leq L \cdot \frac{1}{\deg(v)+1} \left( \| \mathbf{x}_{v'} - \bar{\mathbf{x}}_{\mathcal{N}(v)} \| + |N(v')| \right).
\end{equation}
\qedhere
\end{proof}

\noindent \textbf{Extension to Degree-Normalized GNNs.}
While the above analysis assumes mean aggregation for clarity, similar reasoning applies to degree-normalized GNNs such as GCNs~\cite{kipf2016semi}. 
In GCNs, neighbor features are weighted by $\frac{1}{\sqrt{\deg(v)\deg(u)}}$ during aggregation, or else the node representations are updated via neighborhood aggregation:
\begin{equation}
\mathbf{h}_v = \sigma\left( \sum_{u \in \mathcal{N}(v)} \frac{1}{\sqrt{\deg(v) \deg(u)}} \mathbf{W} \mathbf{x}_u \right),
\end{equation}
where
$\mathcal{N}(v)$ denotes the set of neighbors of $v$,
$\mathbf{W} \in \mathbb{R}^{d' \times d}$ is a learnable weight matrix, and 
$\sigma$ is a non-linear activation function (e.g., ReLU).

Now, suppose a new node $v'$ is injected, with feature vector $\mathbf{x}_{v'}$ and edges $E'$ connecting it to a set of existing nodes $N(v')$.
The perturbed graph is $G'$, with feature matrix $\mathbf{X}'$ and adjacency matrix $\mathbf{A}'$. 
The new aggregation for a neighbor $v$ of $v'$ becomes:
\begin{equation}
\mathbf{h}_v' = \sigma\left( \sum_{u \in \mathcal{N}'(v)} \frac{1}{\sqrt{\deg'(v) \deg'(u)}} \mathbf{W} \mathbf{x}_u' \right),
\end{equation}
where $\mathcal{N}'(v)$ and $\deg'(v)$ reflect the updated graph $G'$. 
Although the exact perturbation expressions involve more complex degree-dependent terms, the overall structure of the bound remains similar: 
the confidence shift $\Delta_{\text{conf}}(v)$ still relies on the injected node's feature discrepancy $\| \mathbf{x}_{v'} - \bar{\mathbf{x}}_{\mathcal{N}(v)} \|$ and the relative change in degree normalization.
Thus, the key insight that both structural and semantic perturbations jointly influence local prediction shifts holds more generally across common GNN architectures.

\noindent \textbf{Intuition.} The result of the lemma formally shows that local prediction shifts caused by an injected node are bounded by two key factors: (1) the feature discrepancy between the injected node and its neighbors, and (2) the structural degree impact from new connections.
Importantly, it quantifies how multi-modal perturbations, which modify both textual features and connectivity, jointly influence downstream predictions, justifying the need for multi-objective attack optimization as used in \method.

\subsection{Search Space: Proof}
Recall that in our threat model (Sec.~\ref{sec:method}), the attacker injects $r$ new nodes 
\(V'\) into the original graph \(G= (V, E, T, Y)\),
selecting: (1)~Edge connections $E' \subseteq V' \times V$ from injected to existing nodes, and (2)~Features $T': V' \to \mathcal{T}$. 
The goal is to maximize the model degradation $\Delta_{\text{acc}}(G')$ under black-box access, where $G' = (V \cup V', E \cup E', T \cup T', Y \cup Y')$.

\textbf{Structure-Only Search Space.}
For structure-only node injection attacks, each injected node $v'_i$ may select up to $d_{\max}$ connections to existing nodes $V$.
The number of possible structural connection patterns is:
\(|\mathcal{A}_{\text{structure}}| = \prod_{i=1}^{r} \binom{|V|}{d_i}\), where \(d_i\) is the degree assigned to injected node \(v'_i\). 
Assuming a uniform budget \(d_i = d_{\max}\) for all \(i\), this simplifies to: \(|\mathcal{A}_{\text{structure}}| = \binom{|V|}{d_{\max}}^r\). 
Thus, the structure-only search space grows approximately as \(O(|V|^{r \cdot d_{\max}})\) when \(d_{\max} \ll |V|\).

\paragraph{Feature-Only Search Space.}
For feature-only node injection attacks, assume the feature generation process selects from a discrete candidate set \(\mathcal{F}\) (e.g., class labels, template embeddings).
The number of possible feature assignments is:
\(|\mathcal{A}_{\text{feature}}| = |\mathcal{F}|^r\). 

\paragraph{Multi-Modal Search Space.}
For joint structure-and-feature optimization, the attack space is the Cartesian product:
\(|\mathcal{A}_{\text{multi}}| = |\mathcal{A}_{\text{structure}}| \times |\mathcal{A}_{\text{feature}}| = \binom{|V|}{d_{\max}}^r \times |\mathcal{F}|^r\).
Thus, the overall search space grows as:
\(|\mathcal{A}_{\text{multi}}| = O\left(|V|^{r \cdot d_{\max}} \times |\mathcal{F}|^r\right)\).

\paragraph{Hardness Implication.}
Even for moderate graph sizes and small injection budgets, exact optimization over \(\mathcal{A}_{\text{multi}}\) is computationally infeasible.
Multi-modal attacks involve simultaneously navigating two discrete, high-dimensional combinatorial spaces (structure and semantics) resulting in a search space that is exponentially larger than that of structure-only or feature-only attacks.

\begin{proof}[Proof Sketch for Lemma~\ref{lem:complexity}]
At each generation, \method{} evaluates $N_p$ candidate injection strategies.
Each evaluation requires $Q_{\text{eval}}$ model queries to compute the fitness objectives:
(i)~querying predictions for roughly $O(r \cdot d_{\max}^2)$ nodes in the two-hop neighborhood, and (ii)~computing PageRank scores, which can be approximated in $O(|E|)$ time.
The evolutionary operations (selection, crossover, mutation) require at most $O(N_p)$ operations per generation.
Thus, the total evaluation cost per generation is $O(N_p \cdot (r \cdot d_{\max}^2 + |E|))$. Repeating over $T_{\text{gen}}$ generations gives the total cost. 
\end{proof}

\section{Additional Empirical Analysis}
\label{app:results}
\subsection{Dataset Statistics}
Table \ref{tab:dataset_stats} reports the statistics for the five text-attributed graph datasets used in our experiments, which we described in Sec. \ref{sec: Empirical_setup}.

\setlength{\tabcolsep}{1mm}
\renewcommand{\arraystretch}{0.85}
\setlength{\extrarowheight}{-1pt}
\begin{table}[!htb]
    \centering
    \scriptsize
    \resizebox{0.5\textwidth}{!}{%
    \begin{tabular}{l r r r }
        \hline
        \textbf{Dataset} & \textbf{Node} & \textbf{Edge} & \textbf{Class}  \\
        \hline
        Cora & 2,708 & 5,429 & 7  \\
        PubMed& 19,717 & 44,335& 3  \\
        WikiCS & 11,701 & 215,863&10  \\
        ogbn-products (subset)& 54,025& 74,420&47 \\
        ogbn-arxiv& 169,343& 1,166,243& 40 \\
 \hline
    \end{tabular}
    }
    \caption{Dataset Statistics}
    \label{tab:dataset_stats}
\end{table}

\subsection{Baseline Details}
\label{app:baselines}

We listed our 8 baselines in Sec.~\ref{sec: Empirical_setup}. Here, we provide more details for each baseline, including implementation details and hyperparameter tuning.

\begin{itemize}
    \item Preferential attack~\cite{barabasi1999emergence}: A heuristics-based approach where the probability of each injected node being connected to any of the existing nodes is proportional to the degree of the node. 
    \item Nettack~\cite{zugner2018adversarial}: 
    A popular targeted modification attack. To adapt it for our attack setting, we sample from existing nodes to create injected nodes, and initialize their connections using preferential attachment. During each injection step, randomly sampled test nodes are selected as the attack target, and Nettack is applied to modify the edges between the target and the injected nodes. Our implementation uses a 2-layer GCN as the surrogate model. 
    
    \item Particle Swarm Optimization (PSO)~\cite{Zang_2020_pso}: A population-based node injection attack where each particle represents a candidate set of edges for the injected nodes. The particles are iteratively updated based on their own best-known positions and the global best solution. Fitness is defined by the post-attack classification accuracy using a 2-layer GCN surrogate, targeting global degradation. The algorithm searches for edge configurations that maximize the overall attack effect. 
    
    \item TDGIA~\cite{zou2021_tdgia}: A global node injection attack that selects adversarial neighbors using a topological defective edge strategy. Since the original method assumes continuous features, we adapt it by assigning feature vectors sampled from existing nodes to ensure consistency.
    
    \item AFGSM~\cite{wang2020scalable_AFGSM}: A gradient-based, targeted attack. We extend it to global poisoning by selecting random test nodes as the target for each injected node and applying AFGSM accordingly. 
    
    \item \(G^{2}\)A2C~\cite{ju2023let_g2a2c}: A reinforcement learning-based, targeted evasion attack that models the process as a Markov decision problem. We adapt it to the global poisoning setting by applying the policy sequentially across injected nodes.
    
    \item  GANI~\cite{fang2024gani}: A global node injection attack that uses a genetic algorithm to select neighbors for each injected node, optimizing for classification loss via a surrogate SGC model and decrease in node homophily. 
    
    \item Word-frequency-based Text-level Graph Injection Attack (WTGIA) \cite{lei2024intruding_WTGIA}: A text-level node injection attack that optimizes sparse binary Bag-of-Words embeddings, then uses an LLM (Llama-2-7b-chat-hf) to generate text containing the selected words. A multi-round correction step ensures word coverage. 
\end{itemize}

\textbf{Implementation Details \& Hyperparameter Tuning.}   
For all target models and directly applicable baselines methods, we use the official implementations and adopt the recommended hyperparameter settings reported in their respective papers, without further tuning. Where necessary, we verify correctness by reproducing baseline performance on each dataset. To ensure fairness, for all attacks, we employ the same method of sampling the edge budgets from the original dataset's distribution.

In our framework, the population size and maximal iteration are set to 30 and 50, respectively. The crossover and mutation probabilities are set to 0.5 and 0.25, respectively. 
We use a range of injection budget, 
corresponding to \{1\%, 2\%, 3\%, 4\%, 5\%\} of the original graph size. The degree of each injected node is sampled from the existing degree distribution.

All experiments were conducted on a machine with an Intel Xeon Gold 6226R CPU, an NVIDIA A40 GPU, and 128GB RAM.

\subsection{Results on Feature-level Enhancer Target Model}

In the main paper, we provided detailed results for a representation-level enhancer target model. Here, we report node classification accuracy after injection attacks on the \textit{feature}-level enhancer target model. This model encodes node features using a pretrained LLM and passes them into a standard GCN, as described in the Empirical Setup (Sec.~\ref{sec: Empirical_setup}). 

\begin{table*}[!htbp]
\centering
\footnotesize
\begin{tabular}{llccccccc}
\toprule
\textbf{Dataset} & \textbf{Methods} & \multicolumn{7}{c}{\textbf{Classification Accuracy}} \\ \cmidrule{3-9}
& & \textbf{Clean} & \textbf{r=0.01} & \textbf{r=0.02} & \textbf{r=0.03} & \textbf{r=0.04} & \textbf{r=0.05} & \textbf{Avg. Rank}\\
\midrule
\multirow{11}{*}{\textbf{Cora}}& Preferential & \multirow{11}{*}{79.98{\scriptsize ±0.72}}
&79.66{\scriptsize ±0.49}
&79.27{\scriptsize ±0.40}
&78.84{\scriptsize ±0.71}
&77.99{\scriptsize ±0.58}
&76.97{\scriptsize ±0.97}
& 11
\\
& Best textual* &  & ----------------- & ----------------- & 73.37{\scriptsize ±0.91} & ----------------- & ----------------- & 7 \\
& PRBCD& 
&78.55{\scriptsize ±0.89}
&76.68{\scriptsize ±0.23}
&75.22{\scriptsize ±0.38}
&74.52{\scriptsize ±0.16}
&73.62{\scriptsize ±0.36}
& 9.8
\\
& Nettack& 
&75.58{\scriptsize ±0.73}
&72.77{\scriptsize ±0.45}
&68.25{\scriptsize ±0.62}
&64.06{\scriptsize ±0.58}
&58.23{\scriptsize ±0.79}
& 5.4\\
& PSO& 
&75.15{\scriptsize ±0.18}
&72.79{\scriptsize ±0.52}
&68.28{\scriptsize ±0.47}
&63.51{\scriptsize ±0.52}
&56.82{\scriptsize ±0.26}
&5\\
& TDGIA& 
&75.54{\scriptsize ±0.29}
&73.49{\scriptsize ±0.78}
&67.49{\scriptsize ±0.68}
&62.51{\scriptsize ±0.26}
&58.59{\scriptsize ±0.67}
&5.4\\
& AFGSM& 
&78.37{\scriptsize ±0.95}
&76.22{\scriptsize ±0.13}
&73.71{\scriptsize ±0.65}
&68.21{\scriptsize ±0.23}
&64.67{\scriptsize ±0.85}
&8.4\\
& $G^{2}A2C$& 
&78.77{\scriptsize ±0.52}
&74.36{\scriptsize ±0.36}
&69.84{\scriptsize ±0.62}
&65.89{\scriptsize ±0.24}
&60.01{\scriptsize ±0.96}
&7.8\\
& GANI& 
&\textbf{72.13}{\scriptsize ±0.53}
&70.63{\scriptsize ±0.82}
&65.66{\scriptsize ±0.69}
&\underline{60.64{\scriptsize ±0.86}}
&56.35{\scriptsize ±0.92}
&2.4\\
& WTGIA& 
&74.57{\scriptsize ±0.83}
&\textbf{69.18}{\scriptsize ±0.22}
&\underline{65.47{\scriptsize ±0.29}}
&60.91{\scriptsize ±0.35}
&\textbf{55.11}{\scriptsize ±0.21}
&2.2\\
& \method  & 
&\underline{72.68{\scriptsize ±0.84}}
&\underline{70.05{\scriptsize ±0.57}}
&\textbf{64.23}{\scriptsize ±0.78}
&\textbf{57.93}{\scriptsize ±0.50}
&\underline{55.64{\scriptsize ±1.08}}
&\textbf{1.6}\\
\midrule
\multirow{11}{*}{\textbf{PubMed}}
& Preferential & \multirow{11}{*}{80.03{\scriptsize ±0.88}}
&77.95{\scriptsize ±0.52}
&75.03{\scriptsize ±0.40}
&74.06{\scriptsize ±0.36}
&72.94{\scriptsize ±0.34}
&68.91{\scriptsize ±0.58}
&10.8\\

& Best textual* &  & ----------------- & -----------------& {73.06{\scriptsize ±0.70}} & ----------------- & ----------------- & 8.4\\

& PRBCD& 
&\textbf{70.74}{\scriptsize ±0.47}
&68.91{\scriptsize ±0.86}
&68.04{\scriptsize ±0.81}
&66.58{\scriptsize ±0.55}
&65.22{\scriptsize ±0.68}
&5.4\\

& Nettack& 
&74.15{\scriptsize ±0.68}
&71.38{\scriptsize ±0.36}
&68.21{\scriptsize ±1.07}
&64.83{\scriptsize ±0.79}
&62.30{\scriptsize ±0.45}
&6.2\\
& PSO& 
&73.28{\scriptsize ±0.12}
&69.81{\scriptsize ±0.63}
&\textbf{66.47}{\scriptsize ±0.50}
&64.64{\scriptsize ±0.53}
&\underline{59.33{\scriptsize ±0.15}}
&3.2\\
& TDGIA& 
&75.57{\scriptsize ±0.77}
&69.61{\scriptsize ±0.62}
&67.99{\scriptsize ±0.51}
&\underline{64.49{\scriptsize ±0.85}}
&61.12{\scriptsize ±0.24}
&4.6\\
& AFGSM& 
&76.02{\scriptsize ±0.75}
&74.48{\scriptsize ±0.48}
&68.58{\scriptsize ±0.66}
&65.82{\scriptsize ±0.68}
&63.00{\scriptsize ±0.43}
&8.4\\
& $G^{2}A2C$& 
&73.61{\scriptsize ±0.18}
&69.74{\scriptsize ±0.95}
&67.72{\scriptsize ±0.83}
&65.98{\scriptsize ±0.37}
&60.56{\scriptsize ±0.40}
&5\\
& GANI& 
&75.57{\scriptsize ±0.91}
&73.10{\scriptsize ±0.55}
&72.58{\scriptsize ±0.90}
&70.28{\scriptsize ±0.66}
&62.85{\scriptsize ±0.61}
&8.4\\
& WTGIA& 
&\underline{72.30{\scriptsize ±0.16}}
&\textbf{68.23}{\scriptsize ±0.87}
&\underline{66.82{\scriptsize ±0.64}}
&65.39{\scriptsize ±0.22}
&61.66{\scriptsize ±0.76}
&3\\
& \method & 
&73.49{\scriptsize ±0.43}
&\underline{68.65{\scriptsize ±0.28}}
&67.62{\scriptsize ±0.57}
&\textbf{64.17}{\scriptsize ±0.37}
&\textbf{58.71}{\scriptsize ±0.56}
&\textbf{2.4}\\
\midrule
\multirow{11}{*}{\textbf{WikiCS}}
& Preferential & \multirow{11}{*}{76.44{\scriptsize ±0.53}}
&75.33{\scriptsize ±0.80}
&74.21{\scriptsize ±0.59}
&72.85{\scriptsize ±0.95}
&72.27{\scriptsize ±0.46}
&71.28{\scriptsize ±0.47}
&11\\

& Best textual* &  & ----------------- & ----------------- & {74.79{\scriptsize ±0.44}} & ----------------- & -----------------&9\\

& PRBCD& 
&73.79{\scriptsize ±0.83}
&71.42{\scriptsize ±0.91}
&69.37{\scriptsize ±0.46}
&69.21{\scriptsize ±0.79}
&67.92{\scriptsize ±0.42}
&9.6\\
& Nettack& 
&74.61{\scriptsize ±1.14}
&73.83{\scriptsize ±0.20}
&73.39{\scriptsize ±0.24}
&72.53{\scriptsize ±0.34}
&71.17{\scriptsize ±0.54}
&2\\
& PSO& 
&73.92{\scriptsize ±0.38}
&71.67{\scriptsize ±0.25}
&71.33{\scriptsize ±0.80}
&69.16{\scriptsize ±0.56}
&67.35{\scriptsize ±0.68}
&4.2\\
& TDGIA& 
&74.17{\scriptsize ±0.63}
&73.08{\scriptsize ±0.36}
&72.39{\scriptsize ±0.28}
&69.86{\scriptsize ±0.77}
&68.22{\scriptsize ±0.45}
&2.6\\
& AFGSM& 
&73.05{\scriptsize ±0.50}
&71.64{\scriptsize ±0.52}
&70.92{\scriptsize ±0.31}
&69.72{\scriptsize ±0.43}
&68.03{\scriptsize ±0.88}
&6\\
& $G^{2}A2C$& 
&74.28{\scriptsize ±0.51}
&72.47{\scriptsize ±0.66}
&71.47{\scriptsize ±0.69}
&70.31{\scriptsize ±0.42}
&68.39{\scriptsize ±0.40}
&7\\
& GANI& 
&73.36{\scriptsize ±0.17}
&72.17{\scriptsize ±0.36}
&71.31{\scriptsize ±0.73}
&70.77{\scriptsize ±0.81}
&70.06{\scriptsize ±0.28}
&4.8\\
& WTGIA& 
&\underline{72.81{\scriptsize ±0.18}}
&\underline{71.33{\scriptsize ±0.37}}
&\textbf{68.37}{\scriptsize ±0.51}
&\underline{67.41{\scriptsize ±0.78}}
&\underline{65.04{\scriptsize ±0.54}}
&8.4\\
& \method & 
&\textbf{72.06}{\scriptsize ±0.32}
&\textbf{70.68}{\scriptsize ±0.49}
&\underline{68.72{\scriptsize ±0.58}}
&\textbf{66.42}{\scriptsize ±0.35}
&\textbf{64.89}{\scriptsize ±0.63}
&\textbf{1.4}\\

\midrule
\multirow{11}{*}{\textbf{ogbn-arxiv}}
& Preferential & \multirow{11}{*}{69.53{\scriptsize ±0.36}}
&69.50{\scriptsize ±0.45}
&69.37{\scriptsize ±0.22}
&68.82{\scriptsize ±0.51}
&68.38{\scriptsize ±0.25}
&67.97{\scriptsize ±0.29}
&11\\

& Best textual* &  & ----------------- & ----------------- & {67.33{\scriptsize ±0.75}} & ----------------- & -----------------&9\\

& PRBCD& 
&68.29{\scriptsize ±0.26}
&67.92{\scriptsize ±0.22}
&67.54{\scriptsize ±0.16}
&66.90{\scriptsize ±0.63}
&66.04{\scriptsize ±0.27}
&9.6\\

& Nettack& 
&\textbf{62.57}{\scriptsize ±0.15}
&\underline{60.83{\scriptsize ±0.34}}
&\underline{59.46{\scriptsize ±0.30}}
&\underline{57.38{\scriptsize ±0.51}}
&56.83{\scriptsize ±0.69}
&2\\
& PSO& 
&64.56{\scriptsize ±0.34}
&63.60{\scriptsize ±0.26}
&62.56{\scriptsize ±0.10}
&61.96{\scriptsize ±0.09}
&61.85{\scriptsize ±0.18}
&4.2\\
& TDGIA& 
&\underline{62.84{\scriptsize ±0.16}}
&60.84{\scriptsize ±0.41}
&59.48{\scriptsize ±0.35}
&57.93{\scriptsize ±0.24}
&\underline{56.65{\scriptsize ±0.22}}
&2.6\\
& AFGSM& 
&64.78{\scriptsize ±0.39}
&63.96{\scriptsize ±0.53}
&63.81{\scriptsize ±0.78}
&63.11{\scriptsize ±0.40}
&62.69{\scriptsize ±0.17}
&6.2\\
& $G^{2}A2C$& 
&65.54{\scriptsize ±0.47}
&64.89{\scriptsize ±0.14}
&64.44{\scriptsize ±0.34}
&63.70{\scriptsize ±0.38}
&63.32{\scriptsize ±0.16}
&6.4\\
& GANI& 
&64.28{\scriptsize ±0.76}
&63.92{\scriptsize ±0.40}
&63.42{\scriptsize ±0.37}
&62.49{\scriptsize ±0.43}
&61.92{\scriptsize ±0.11}
&5.2\\
& WTGIA& 
&68.21{\scriptsize ±0.51}
&67.61{\scriptsize ±0.39}
&66.96{\scriptsize ±0.25}
&66.11{\scriptsize ±0.33}
&65.60{\scriptsize ±0.52}
&8.4\\
& \method & 
&63.31{\scriptsize ±0.29}
&\textbf{59.97}{\scriptsize ±0.74}
&\textbf{58.31}{\scriptsize ±0.68}
&\textbf{57.35}{\scriptsize ±0.52}
&\textbf{56.46}{\scriptsize ±0.37}
&\textbf{1.4}\\

\midrule
\multirow{11}{*}{\textbf{ogbn-products}}
& Preferential & \multirow{11}{*}{82.86{\scriptsize ±0.50}}
&79.21{\scriptsize ±0.24}
&77.86{\scriptsize ±0.14}
&76.89{\scriptsize ±0.68}
&75.91{\scriptsize ±0.24}
&74.60{\scriptsize ±0.73}
&7.6\\

& Best textual* &  & ----------------- & -----------------& {79.44{\scriptsize ±0.36}} & ----------------- & -----------------&10.6\\

& PRBCD& 
&80.17{\scriptsize ±0.18}
&78.82{\scriptsize ±0.46}
&78.02{\scriptsize ±0.35}
&77.44{\scriptsize ±0.64}
&75.46{\scriptsize ±0.27}
&10.2\\

& Nettack& 
&79.42{\scriptsize ±0.38}
&78.52{\scriptsize ±0.12}
&76.77{\scriptsize ±0.34}
&75.63{\scriptsize ±0.23}
&74.35{\scriptsize ±0.55}
&7.6\\
& PSO& 
&79.05{\scriptsize ±0.22}
&77.39{\scriptsize ±0.27}
&76.04{\scriptsize ±0.21}
&74.92{\scriptsize ±0.41}
&74.12{\scriptsize ±0.94}
&4\\
& TDGIA& 
&79.46{\scriptsize ±0.17}
&78.67{\scriptsize ±0.29}
&76.20{\scriptsize ±0.44}
&75.54{\scriptsize ±0.56}
&74.41{\scriptsize ±0.23}
&7.8\\
& AFGSM& 
&79.16{\scriptsize ±0.36}
&78.52{\scriptsize ±0.22}
&76.75{\scriptsize ±0.18}
&75.24{\scriptsize ±0.29}
&74.13{\scriptsize ±0.51}
&6.2\\
& $G^{2}A2C$& 
&79.31{\scriptsize ±0.28}
&78.35{\scriptsize ±0.43}
&76.54{\scriptsize ±0.66}
&73.79{\scriptsize ±0.81}
&73.54{\scriptsize ±0.19}
&5.2\\
& GANI& 
&\textbf{78.99}{\scriptsize ±0.21}
&\underline{77.25{\scriptsize ±0.73}}
&76.16{\scriptsize ±0.35}
&\textbf{73.47}{\scriptsize ±0.16}
&72.90{\scriptsize ±0.30}
&2.2\\
& WTGIA& 
&79.06{\scriptsize ±0.24}
&77.28{\scriptsize ±0.60}
&\underline{75.52{\scriptsize ±0.84}}
&\underline{73.59{\scriptsize ±0.78}}
&\textbf{72.68}{\scriptsize ±0.62}
&2.4\\
& \method & 
&\underline{79.01{\scriptsize ±0.14}}
&\textbf{76.77}{\scriptsize ±0.59}
&\textbf{75.33}{\scriptsize ±0.81}
&73.86{\scriptsize ±0.75}
&\underline{72.85{\scriptsize ±0.19}}
&\textbf{2}\\

\bottomrule
\end{tabular}%
\caption{[Feature-level enhancer model] Classification accuracy after injection attacks. Best performance (lowest accuracy) bolded, second best underlined. \method{} maintains superior attack performance across datasets and budgets. }
\label{tab:attack_results_feat}
\end{table*}

Table~\ref{tab:attack_results_feat} presents post-attack accuracy under varying injection budgets, datasets, and attack methods. The last column summarizes the average rank of each attack across the 5 budgets $r$, where lower number rank corresponds to a more effective attack.
The results are consistent with our findings in Sec. \ref{sec:overall_performance}, \method{} achieves the best attack performance on average. 

\subsection{Additional Runtime Results} 

Figure~\ref{fig:runtime} shows the runtime results across all datasets. For each method, we report the average wall-clock time per node injection. The observations are consistent with Sec. \ref{sec:runtime_analysis}, \method{} achieves the lowest average runtime across datasets.

\begin{figure}[!htbp]
  \centering
  \includegraphics[width=0.95\linewidth]{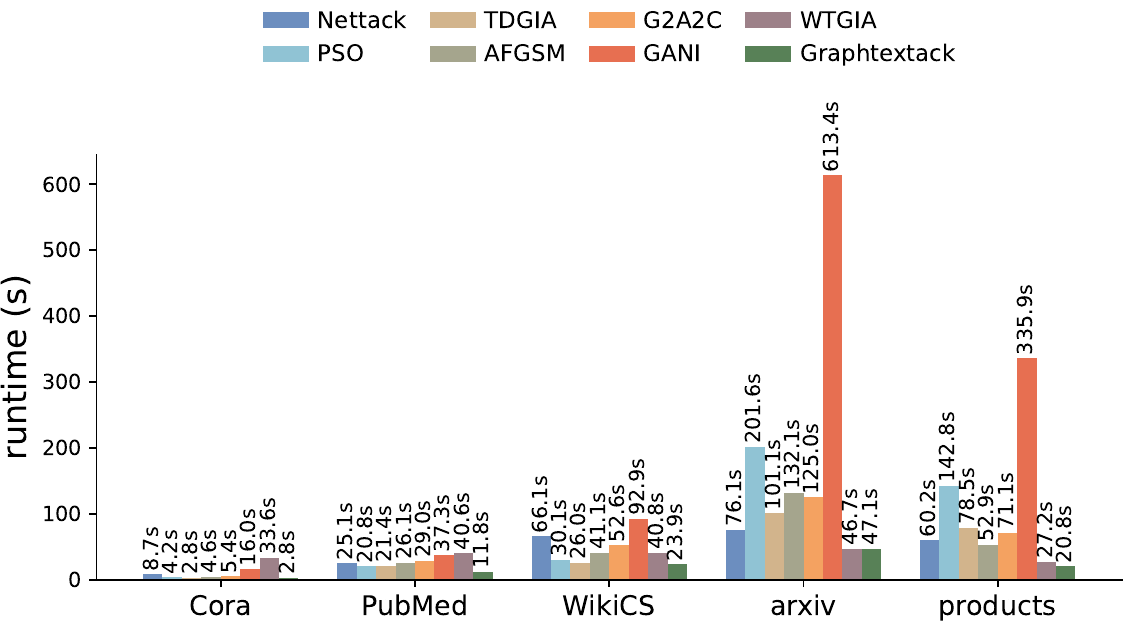}
  \caption{[Extended results] Comparison of runtime to generate an injection on the representation-level enhancer target model.}
  \label{fig:runtime}
\end{figure}

\subsection{Additional Ablation Study Results} 
\label{app:ablation}

In table \ref{tab:ablation_all}, we provide full ablation results for \method{} across three datasets (Cora, WikiCS, and ogbn-arxiv) under varying injection budgets. These results are consistent with our findings reported in Sec. \ref{sec:ablation}. Each variant disables a key component of the framework: crossover, mutation, or one of the fitness terms. Removing any component consistently reduces attack effectiveness, highlighting the importance of both population-level diversity and multi-modal fitness objectives.

\begin{table}[!htbp]
    \centering
    \scriptsize
    \begin{tabular}{lccccc}
        \toprule
        \multicolumn{6}{c}{\textbf{Cora}} \\
        \midrule
        \textbf{Variants} & \textbf{r=0.01} & \textbf{r=0.02} & \textbf{r=0.03} & \textbf{r=0.04} & \textbf{r=0.05} \\ 
        \midrule
        \method 
        & \textbf{73.99}{\scriptsize ±0.78} 
        & \underline{71.14{\scriptsize ±0.68}} 
        & \textbf{65.75}{\scriptsize ±0.81} 
        & \underline{64.96{\scriptsize ±0.63}} 
        & \textbf{62.02}{\scriptsize ±0.64} \\
        Without crossover 
        & 75.11{\scriptsize ±0.17} 
        & 72.79{\scriptsize ±0.44} 
        & 67.82{\scriptsize ±0.36} 
        & 65.49{\scriptsize ±0.29} 
        & 63.35{\scriptsize ±0.62} \\
        Without mutation & 76.29{\scriptsize ±0.67} 
        & 73.19{\scriptsize ±0.30} 
        & 68.47{\scriptsize ±0.47} 
        & 65.80{\scriptsize ±0.22} 
        & 63.25{\scriptsize ±0.24} \\
        Without pred. shift 
        & 76.60{\scriptsize ±0.54} 
        & 71.95{\scriptsize ±0.20}
        & 71.56{\scriptsize ±0.52} 
        & 66.55{\scriptsize ±0.26} 
        & 63.49{\scriptsize ±0.76} \\
        Without PageRank 
        & \underline{74.71{\scriptsize ±0.72}} 
        & \textbf{70.94}{\scriptsize ±0.53} 
        & \underline{68.90{\scriptsize ±0.64}} 
        & \textbf{64.70}{\scriptsize ±0.49} 
        & \underline{64.01{\scriptsize ±0.50}} \\
        
        \midrule
        \multicolumn{6}{c}{\textbf{WikiCS}} \\
        \midrule
        \textbf{Variants} & \textbf{r=0.01} & \textbf{r=0.02} & \textbf{r=0.03} & \textbf{r=0.04} & \textbf{r=0.05} \\ 
        \midrule
        \method 
        & \textbf{71.69}{\scriptsize ±1.04} 
        & \textbf{67.77}{\scriptsize ±1.27} 
        & \textbf{64.58}{\scriptsize ±0.69} 
        & \textbf{63.56}{\scriptsize ±1.36} 
        & \textbf{61.35}{\scriptsize ±1.33} \\
        Without crossover
        &73.50{\scriptsize ±0.71} 
        &71.19{\scriptsize ±0.14} 
        &67.60{\scriptsize ±0.42} 
        &66.01{\scriptsize ±0.97} 
        &64.42{\scriptsize ±0.59} \\
        Without mutation 
        & 73.39{\scriptsize ±0.85} 
        & 70.57{\scriptsize ±0.22} 
        & 68.89{\scriptsize ±0.31} 
        & 66.88{\scriptsize ±0.39} 
        & 65.19{\scriptsize ±0.41} \\
        Without pred. shift 
        & 73.77{\scriptsize ±0.67} 
        & 72.65{\scriptsize ±0.35} 
        & 71.14{\scriptsize ±0.34} 
        & 70.85{\scriptsize ±0.61} 
        & 66.83{\scriptsize ±0.54} \\
        Without PageRank 
        & \underline{71.92{\scriptsize ±0.78}} 
        & \underline{68.84{\scriptsize ±0.76}} 
        & \underline{67.53{\scriptsize ±0.33}} 
        & \underline{64.44{\scriptsize ±0.38}} 
        & \underline{61.97{\scriptsize ±1.00}} \\

        \midrule
        \multicolumn{6}{c}{\textbf{ogbn-arxiv}} \\
        \midrule
        \textbf{Variants} & \textbf{r=0.01} & \textbf{r=0.02} & \textbf{r=0.03} & \textbf{r=0.04} & \textbf{r=0.05} \\ 
        \midrule
        \method 
        & \textbf{71.95}{\scriptsize ±0.38} 
        & \underline{70.86{\scriptsize ±0.89}} 
        & \textbf{68.23}{\scriptsize ±0.85} 
        & \textbf{67.95}{\scriptsize ±0.95} 
        & \textbf{66.61}{\scriptsize ±0.78} \\
        Without crossover 
        & 73.43{\scriptsize ±0.52} 
        & 72.69{\scriptsize ±0.33} 
        & 69.47{\scriptsize ±0.64} 
        & 68.17{\scriptsize ±0.30} 
        & 67.85{\scriptsize ±0.43} \\
        Without mutation 
        & 72.32{\scriptsize ±0.26} 
        & 72.54{\scriptsize ±0.24} 
        & 71.19{\scriptsize ±0.35} 
        & 70.85{\scriptsize ±0.16} 
        & 69.75{\scriptsize ±0.57} \\
        Without pred. shift 
        & 74.52{\scriptsize ±0.23} 
        & 74.26{\scriptsize ±0.45} 
        & 72.17{\scriptsize ±0.10} 
        & 71.39{\scriptsize ±0.72} 
        & 70.89{\scriptsize ±0.63} \\
        Without PageRank 
        & \underline{71.97{\scriptsize ±0.47}} 
        & \textbf{70.55}{\scriptsize ±0.19} 
        & \underline{68.94{\scriptsize ±0.12}}
        & \underline{68.41{\scriptsize ±0.18}} 
        & \underline{67.04{\scriptsize ±0.36}} \\

        \bottomrule
    \end{tabular}
    \caption{Ablation study for \method{} on Cora, WikiCS, and ogbn-arxiv: Accuracy after injection under different model configurations on the representation-level enhancer model.}
    \label{tab:ablation_all}
\end{table}

\end{document}